\makeatletter
\def\input@path{{template/}}
\makeatother

\documentclass[runningheads]{llncs}

\newif\ifpublish
\newif\ifextended %
\publishtrue
\extendedtrue

\usepackage{graphicx}
\usepackage{xcolor}
\usepackage{listings}
\usepackage{amsmath}
\usepackage{xparse}
\usepackage{pgfplots}
\pgfplotsset{compat=1.17}
\usepackage{tikz}
\usetikzlibrary{positioning, shapes.geometric, arrows.meta, calc}
\usepackage{booktabs}
\usepackage{etoolbox}
\usepackage{xspace}
\usepackage{enumitem}
\setlist{nosep, leftmargin=*}
\usepackage{hyperref}
\usepackage{cleveref}

\hypersetup{
	colorlinks=true,
	linkcolor={magenta!80!black},
	citecolor={green!55!black},
	urlcolor={black!55},
	pdfborder={0 0 0},
}
\crefname{section}{Section}{Sections}
\Crefname{section}{Section}{Sections}
\crefname{subsection}{Section}{Sections}
\Crefname{subsection}{Section}{Sections}
\crefname{appendix}{Appendix}{Appendices}
\Crefname{appendix}{Appendix}{Appendices}
\crefname{figure}{Figure}{Figures}
\Crefname{figure}{Figure}{Figures}
\crefname{table}{Table}{Tables}
\Crefname{table}{Table}{Tables}
\crefname{equation}{Equation}{Equations}
\Crefname{equation}{Equation}{Equations}
\crefname{theorem}{Theorem}{Theorems}
\Crefname{theorem}{Theorem}{Theorems}
\crefname{lemma}{Lemma}{Lemmas}
\Crefname{lemma}{Lemma}{Lemmas}
\crefname{definition}{Definition}{Definitions}
\Crefname{definition}{Definition}{Definitions}

\newif \ifcomments

\commentsfalse

\ifcomments
	\newcommand{\deepak}[1]{{\color{olive} \textbf{Deepak:} #1}}
	\newcommand{\karl}[1]{{\color{teal} \textbf{Karl:} #1}}
	\newcommand{\alberto}[1]{{\color{blue} \textbf{Alberto:} #1}}
	\newcommand{\george}[1]{{\color{red} \textbf{George:} #1}}
	\newcommand{\todo}[1]{{\color{red} \textbf{TODO:} #1}}
	\newcommand{\arnab}[1]{{\color{magenta} \textbf{Arnab:} #1}}
	\newcommand{\remove}[1]{{\color{red!50!black} \textbf{remove:} #1}}
\else
	\newcommand{\deepak}[1]{}
	\newcommand{\karl}[1]{}
	\newcommand{\alberto}[1]{}
	\newcommand{\george}[1]{}
	\newcommand{\todo}[1]{}
	\newcommand{\arnab}[1]{}
	\newcommand{\remove}[1]{}
\fi

\makeatletter
\patchcmd{\@makecaption}{\small}{\scriptsize}{}{\ClassError{macros}{caption patch failed}{}}
\patchcmd{\table}{\setlength\belowcaptionskip{10\p@}}{\setlength\belowcaptionskip{4\p@}}{}{\ClassError{macros}{table caption patch failed}{}}
\makeatother
\newcommand{\sysname}{Guppy}

\newcommand{\codelink}{
	\ifpublish
		\url{https://github.com/asonnino/guppy} (commit \texttt{9c3075f})
	\else
		\url{https://anonymous.4open.science/r/minnow}
	\fi
}

\newcommand{\stream}{S}
\newcommand{\streamid}{id}
\newcommand{\streamupdate}{p}
\newcommand{\checkpointnumber}{c}
\newcommand{\cp}{\checkpointnumber}
\newcommand{\streamstate}{v}

\newcommand{\StreamUpdateFn}{A}

\newcommand{\streampoint}{p}

\newcommand{\localcommit}{{\sf localCommit}}
\newcommand{\genheader}{{\sf genHeader}}
\newcommand{\processcp}{{\sf processUpd}}

\newcommand{\localdigest}[1]{d_{#1}}
\NewDocumentCommand{\globaldigest}{g}{%
	D\IfValueT{#1}{_{#1}}
}
\newcommand{\updateproof}[1]{\Pi_{#1}}
\newcommand{\validityproof}{\pi}

\newcommand{\header}[1]{\mathcal{H}_{#1}}
\NewDocumentCommand{\checkpointsign}{g}{%
	\sigma\IfValueT{#1}{_{#1}}%
}
\NewDocumentCommand{\contents}{g}{%
	\mathcal{C}\IfValueT{#1}{_{#1}}%
}

\newcommand{\pk}{{\sf pk}}
\newcommand{\validatorpk}{\pk_{\sf chain}}

\newcommand{\batchsize}{b}
\newcommand{\bundlesize}{N}
\newcommand{\numbatches}{n}

\newcommand{\batch}[1]{B_{#1}}
\newcommand{\bundle}[1]{U_{#1}}

\newcommand{\chaintag}{{\sf bc}}
\newcommand{\circuittag}{{\sf cir}}

\newcommand{\minnowhash}{H_{\sf \MakeLowercase{\sysname}}}

\newcommand{\chainbundle}[1]{\bundle{#1}^\chaintag}
\newcommand{\chainupdates}[1]{\chainbundle{#1}}
\newcommand{\circuitbundle}[1]{\bundle{#1}^\circuittag}

\newcommand{\numcircuitbatches}{\numbatches}
\newcommand{\numcircuitupdates}{\bundlesize}
\newcommand{\numchainupdates}[1]{\bundlesize_{#1}^\chaintag}
\newcommand{\hashprotocol}{\sysname-Sparse}
\newcommand{\mtprotocol}{\sysname-A}
\newcommand{\minnowMTA}{\mtprotocol}
\newcommand{\minnowmtatime}{t^A}

\newcommand{\latency}{t_{\sf e2e}}
\newcommand{\parallelphasetag}{{\sf par}}
\newcommand{\aggregationphasetag}{{\sf agg}}
\newcommand{\sequentialphasetag}{{\sf seq}}

\newcommand{\numbatchespersubbundle}{\numbatches}

\newcommand{\numupdatespersubbundle}{y}
\newcommand{\subbundlesize}{\numupdatespersubbundle}
\newcommand{\numsubbundles}[1]{\beta_{#1}}

\newcommand{\aggregationfactor}{\alpha}
\newcommand{\aggregationlevels}{l}

\newcommand{\minnowMTB}{\sysname-B}

\newcommand{\parallelphasestatement}{S_\parallelphasetag}
\newcommand{\parallelphaseproof}[1]{\Pi_\parallelphasetag^{(#1)}}

\newcommand{\parallelphasetime}{t_\parallelphasetag}
\newcommand{\parallelphasebatchsize}{y}

\newcommand{\sequentialphasestatement}{S_\sequentialphasetag}
\newcommand{\sequentialphaseproof}[1]{\Pi_\sequentialphasetag^{#1}}
\newcommand{\sequentialphasepubinputs}[1]{pub_\sequentialphasetag^{#1}}
\newcommand{\sequentialphasetime}{t_\sequentialphasetag}

\newcommand{\intermediatephasestatement}{S_\aggregationphasetag}
\newcommand{\intermediatephaseproof}[2]{\Pi_\aggregationphasetag^{(#1)\rightarrow(#2)}}
\newcommand{\intermediatephasepubinputs}[2]{pub_\aggregationphasetag^{(#1)\rightarrow(#2)}}
\newcommand{\intermediatephasetime}{t_\aggregationphasetag}

\newcommand{\treedepth}{d}
\newcommand{\merklepath}{path}
\newcommand{\newfield}[1]{\localdigest{#1}}
\newcommand{\pubinputs}[1]{pub^{#1}}
\newcommand{\streamidset}{\mathcal{ID}}
\newcommand{\globalstatemap}[1]{V^{#1}}

\newcommand{\throughput}{\rho}
\newcommand{\streamnewbit}{e}
\newcommand{\numleaves}[1]{L_{#1}}
\newcommand{\minnowMTAstatement}{S_A}

\newcommand{\bundletime}{t_{\sf bundle}}

\newcommand{\maxnumstreams}{M}

\title{\sysname: Efficient Light Clients via Recursive Zero-Knowledge Proofs}

\ifpublish
    \author{
        George Danezis\inst{1,2} \and
        Deepak Maram\inst{1} \and
        Arnab Roy\inst{1} \and
        Alberto Sonnino\inst{1,2} \and
        Karl W\"ust\inst{1}
    }
    \institute{Mysten Labs \and University College London}
    \authorrunning{G. Danezis et al.}
\else
    \author{}
    \institute{}
\fi

\begin{document}

\maketitle
\ifpublish\else
    \pagestyle{plain}
    \thispagestyle{plain}
\fi

\begin{abstract}
  Traditional light clients rely on validators committing to the
  \emph{entire blockchain state} at every block via a state commitment such
  as a Merkle tree,
  allowing clients to verify facts using short proofs.
  However, maintaining large and ever-growing state trees imposes a
  significant burden on validators and lies on the critical path of
  block production.
  As a result, many modern high-throughput chains avoid this approach
  altogether.
  This work asks whether efficient inclusion proofs can be supported
  \emph{without requiring
    validators to maintain full state commitments}.
  We present \sysname{}, a protocol that achieves this by having validators
  commit to just the \emph{state updates}.
  An off-chain, untrusted service, secured by recursive
  Zero-Knowledge Proofs (ZKPs), then maintains a verifiable Merkle tree
  over the full state.
  This design keeps validator overhead negligible and does not increase
  the asymptotic complexity of block construction.
  Our design rests on two key technical ideas.
  First, a hash-chain commitment moves validator signature verification
  out of the ZK circuit, keeping the proving circuit efficient.
  Second, we design a parallel recursive proving pipeline that
  leverages cheap recursion in modern ZKPs to ensure latency grows only
  logarithmically with throughput.
  Our Plonky2-based implementation demonstrates that \sysname{} can
  maintain a Merkle tree of size~$2^{30}$ while processing thousands of updates
  per second, adding only~2-4\,s of latency.
\end{abstract}

\section{Introduction}
\label{sec:intro}

Light clients let resource-constrained devices, such as wallets on phones, IoT devices, or smart contracts on other chains, verify facts about the chain with short cryptographic proofs instead of downloading and executing the entire ledger. The common design has validators commit to the entire state, of size $M$, with a state commitment such as a Merkle tree of depth $\log M$, as Ethereum~\cite{ethereum} does. Maintaining that commitment is costly because each update touches multiple database entries, the state only grows over time, and the computation sits on the critical path to finality. High-throughput chains~\cite{sui-lutris,solana} process thousands of transactions per second and therefore do not maintain a full-state commitment at all. Their light clients are left without a way to verify state.

Sunfish~\cite{sunfish} shows that a weaker query is cheap to serve: the state of a value at the checkpoint where it was last modified (we call each block a checkpoint). What applications need is completeness, the state of any object at any checkpoint, for example a wallet's current balance or a dapp's package state. This raises the question: \emph{can a blockchain support efficient light clients with completeness without requiring validators to commit to the entire state?} Efficient means that proof sizes remain logarithmic in the system parameters.

\sysname{} answers it by having validators commit only to the state updates of each checkpoint, not to the state. The difference is large: Sui holds roughly $2^{30}$ objects but updates at most about 10k of them per second (about $0.001\%$). An untrusted off-chain ZK service, akin to a ZK co-processor~\cite{axiom,brevis,lagrange-coprocessor}, maintains a Merkle tree over the entire state and, after every checkpoint, proves with a recursive zero-knowledge proof~\cite{recursion-cycles,bitansky2013} that the new root is the correct successor of the previous one under the checkpoint's updates. The straightforward recursive statement verifies the previous proof, verifies the checkpoint under the validator keys, and applies the updates to the tree. Plonky2~\cite{plonky2} makes such recursion practical, at about 6k gates and 0.3\,s to prove. A client asking for an address's balance at checkpoint $\cp$ receives a Merkle path and the proof, even if the address last changed years ago. Because committee key rotations are infrequent (about daily on Ethereum and Sui), clients can be expected to know the validator keys.

Instantiating this design on a high-throughput chain faces three obstacles: validator signature schemes such as BLS are not ZK-friendly, extracting the relevant fields (say balances) may require decoding the entire checkpoint, and the proof must complete before the next checkpoint arrives because the previous proof is one of its inputs, which is hard at thousands of updates per second.

\minnowMTA{} (\Cref{sec:minnow-mt-a}) removes the first two obstacles with one new header field. Suppose checkpoint $\cp$ sets the balances of addresses {\tt 0x123} and {\tt 0x234} to 30 and 100, and let $\chainupdates{\cp}$ be the list of such updates. Each validator computes the hash-chain head $\localdigest{\cp} = H(\localdigest{\cp-1}, \chainupdates{\cp})$ with a ZK-friendly hash (Poseidon), in practice over fixed-size batches, and puts $\localdigest{\cp}$ in the header of checkpoint $\cp$. The circuit no longer verifies checkpoint validity; it exposes the chain head as a public input, and the client checks that the head appears in the signed header (a few hundred bytes) and verifies the validator signatures. Checking the head at checkpoint $\cp$ attests to every earlier checkpoint because heads are chained. The third obstacle remains: every proof must finish before the next checkpoint arrives, $\minnowmtatime \leq \bundletime$ (middle row of \Cref{fig:minnow-protocols}); our Plonky2 benchmarks (\Cref{sec:eval}) cap \minnowMTA{} at about 290\,updates/s for a tree of $2^{30}$ streams, beyond which latency grows without bound.

\begin{figure}[t]
  \centering
  \begin{tikzpicture}[
      >=Stealth,
      font=\scriptsize,
      cp/.style={rectangle, draw, minimum width=4.7cm, minimum height=0.45cm, align=center},
      prover/.style={rectangle, draw, minimum width=1.0cm, minimum height=0.42cm, align=center, fill=black!4},
      wide/.style={rectangle, draw, minimum width=2.2cm, minimum height=0.42cm, align=center, fill=black!4},
      rowlabel/.style={font=\scriptsize, anchor=east},
      feed/.style={->, dashed, black!60},
    ]
    \node[cp] (A) at (0,0) {Checkpoint $i$};
    \node[cp, right=0cm of A] (B) {Checkpoint $i+1$};
    \node[right=0.15cm of B] (C) {$\cdots$};
    \draw[<->] ([yshift=0.1cm]A.north west) -- node[above, inner sep=1pt] {$\bundletime$} ([yshift=0.1cm]A.north east);
    \draw[<->] ([yshift=0.1cm]B.north west) -- node[above, inner sep=1pt] {$\bundletime$} ([yshift=0.1cm]B.north east);
    \node[rowlabel] at ([xshift=-0.2cm]A.west) {Blockchain};

    \node[prover] (E) at ($(A.center)+(0.6,-1.2)$) {$\minnowMTAstatement$};
    \node[prover] (F) at ($(B.center)+(0.6,-1.2)$) {$\minnowMTAstatement$};
    \draw[->] (E) -- node[above, inner sep=1pt] {$\updateproof{i}$} (F);
    \draw[->] (F.east) -- node[above, inner sep=1pt] {$\updateproof{i+1}$} ++(1.6,0);
    \draw[<->] ([yshift=0.08cm]E.north west) -- node[above, inner sep=1pt] {$\minnowmtatime$} ([yshift=0.08cm]E.north east);
    \node[rowlabel] at ([xshift=-0.2cm]A.west |- E) {\minnowMTA{}};

    \node[wide] (H) at ($(A.center)+(-0.95,-2.35)$) {$\parallelphasestatement$};
    \node[prover] (K) at ($(A.center)+(1.35,-2.35)$) {$\sequentialphasestatement$};
    \draw[->] (H) -- node[above, inner sep=1pt] {$\Pi^{\mathit{tmp}}_i$} (K);
    \node[wide] (I) at ($(B.center)+(-0.95,-2.35)$) {$\parallelphasestatement$};
    \node[prover] (L) at ($(B.center)+(1.35,-2.35)$) {$\sequentialphasestatement$};
    \draw[->] (I) -- node[above, inner sep=1pt] {$\Pi^{\mathit{tmp}}_{i+1}$} (L);
    \draw[->] (K) -- ++(0,-0.32) -| node[below, pos=0.25, inner sep=1pt] {$\updateproof{i}$} (L);
    \draw[->] (L.east) -- node[above, inner sep=1pt] {$\updateproof{i+1}$} ++(1.1,0);
    \draw[<->] ([yshift=0.08cm]H.north west) -- node[above, inner sep=1pt] {$\parallelphasetime^B$} ([yshift=0.08cm]H.north east);
    \draw[<->] ([yshift=0.08cm]K.north west) -- node[above, inner sep=1pt] {$\sequentialphasetime^B$} ([yshift=0.08cm]K.north east);
    \node[rowlabel] at ([xshift=-0.2cm]A.west |- H) {\minnowMTB{}};

    \draw[feed] ([xshift=-0.6cm]A.south) -- node[right, pos=0.3, inner sep=1pt] {$\chainupdates{i}$} (E.north west);
    \draw[feed] ([xshift=-1.9cm]A.south) -- node[left, pos=0.82, inner sep=1pt] {$\chainupdates{i}$} (H.north west);
    \draw[feed] ([xshift=-0.6cm]B.south) -- node[right, pos=0.3, inner sep=1pt] {$\chainupdates{i+1}$} (F.north west);
    \draw[feed] ([xshift=-1.9cm]B.south) -- node[left, pos=0.82, inner sep=1pt] {$\chainupdates{i+1}$} (I.north west);

    \draw[dotted, black!50] ($(A.west |- E.north)+(-1.7,0.33)$) -- ++(12.0,0);
    \draw[dotted, black!50] ($(A.west |- H.north)+(-1.7,0.33)$) -- ++(12.0,0);
  \end{tikzpicture}
  \caption{The \sysname{} protocols. The top row shows a blockchain producing
    update bundles. The middle row depicts \minnowMTA{}, which takes
    $\minnowmtatime$ to prove a bundle. The bottom row depicts a simple
    version of \minnowMTB{} with two phases, which take
    $\parallelphasetime^B$ and $\sequentialphasetime^B$ for the parallel and
    sequential phases respectively. \minnowMTA{} requires
    $\minnowmtatime \leq \bundletime$ whereas \minnowMTB{} requires
    $\sequentialphasetime^B \leq \bundletime$.}
  \label{fig:minnow-protocols}
\end{figure}
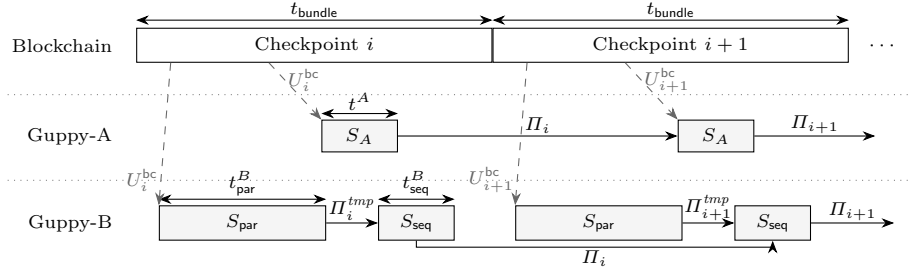

\minnowMTB{} (\Cref{sec:minnow-mt-b}) starts from the observation that the sequential dependency is inherent to incrementally verifiable computation (IVC)~\cite{valiant2008incrementally} but can be confined to a small recursive step. It splits the statement: non-recursive base proofs, each covering a sub-bundle of updates, are produced in parallel on many machines (the boxes labeled $\parallelphasestatement$ in the bottom row of \Cref{fig:minnow-protocols}), an aggregation tree combines them, and a small cyclic proof ($\sequentialphasestatement$) verifies the previous cyclic proof and the aggregated proof, once per checkpoint. Only the cyclic proof, whose cost is constant (0.63\,s), must fit within the checkpoint interval, so \minnowMTB{} sustains any throughput given enough machines. Because aggregation is a tree, end-to-end latency grows only logarithmically with throughput.

We implement both protocols in Rust on Plonky2. The validator-side commitment for $2^{13}$ updates takes under 10\,ms. A one-month measurement of Sui shows that a light-client service must sustain thousands of updates per second at peak. On a distributed testbed of up to 15 AWS machines, \minnowMTB{} with 14 machines processes 2{,}250\,updates/s at a median latency of 4\,s; microbenchmarks predict 8{,}571\,updates/s at 3.15\,s with 41 machines (\Cref{sec:eval}). \sysname{} also broadens what light clients can ask: the same mechanism authenticates streams of events, balances, or objects (\Cref{sec:model}), including queries such as all events of a given type at a checkpoint, which Ethereum's state tree cannot answer.

\paragraph{Contributions.} We make the following contributions:
\begin{itemize}
    \item We present \sysname{}, the first protocol giving light clients completeness (current and historical state of any supported stream) with negligible validator overhead by adding one header field per checkpoint.
    \item We introduce the techniques that make \sysname{} scale: hash-chain commitments that move signature checks out of the circuit, parallel base proofs, tree aggregation, and a constant-cost cyclic proof, giving latency logarithmic in throughput; many apply to other IVC-style computations.
    \item We implement \sysname{} and show thousands of updates per second at a few seconds of latency.
\end{itemize}

\section{Model}
\label{sec:model}

\subsection{Blockchain model}

We model the blockchain as a sequence of \emph{checkpoints} (or \emph{blocks}), where each checkpoint is composed of \emph{contents}, a \emph{header}, and a \emph{signature} from validators. The checkpoint contents $\contents{}$ consist of a list of transactions along with their execution effects, $\{ {\sf tx}_1, {\sf tx}_2, \ldots \}$, where each ${\sf tx}_i$ contains both the original transaction and its \emph{effects}~\cite{sui-lutris}. Effects can contain information such as the outcome of the transaction's execution, its impact on the blockchain's state, and the emitted events~\cite{solidity-events}. The header $\header{}$ includes a checkpoint sequence number~$\checkpointnumber$ unique to each checkpoint, and most commonly it also includes a commitment to the checkpoint contents; we will propose adding one field to the header.

Let $\validatorpk$ be the current committee's public key. The validators in the committee jointly sign the checkpoint header to produce a signature~$\checkpointsign{}$ that can be verified by running ${\sf Sig.verify}(\checkpointsign{}, \header{}, \validatorpk)$ using an unforgeable signature scheme (EUF-CMA). Committee handoffs are authenticated at special checkpoints and the genesis committee is public, allowing anyone to verify a committee by following the sequence from genesis. We denote the contents and header of checkpoint $c$ by $\contents{c}$ and $\header{c}$ respectively, assuming the presence of functions ${\sf IsValid}(\contents{c})$ and ${\sf IsValid}(\header{c})$ that verify their validity using the committee keys. Our model applies to any chain that has signed-checkpoint semantics, such as Ethereum~\cite{ethereum} and Sui~\cite{sui-lutris}. Ethereum uses an account model~\cite{ethereum} where validators commit to a Merkle tree root of all accounts and contracts at the end of every checkpoint, whereas Sui uses an object model supported by a key-value store without a state commitment to maintain high throughput.

\subsection{Streams}

We model blockchain data as a collection of \emph{streams}. A stream~$\stream$ represents a logical view of blockchain state evolution, defined as a key-value map from unique \emph{stream identifiers} (denoted $\streamid$) to short \emph{stream states} (denoted $\streamstate_{\streamid}$) or commitments thereof, i.e., $\stream = \{ (\streamid_1, \streamstate_{\streamid_1}), \ldots, (\streamid_M, \streamstate_{\streamid_M}) \}$. Each stream evolves deterministically as new checkpoints are processed using a function $\StreamUpdateFn$, which derives all stream updates from the current stream and a checkpoint's contents:
\[
  \chainupdates{\cp} \gets
  \StreamUpdateFn(\stream, \contents{\cp}).
\]
The output $\chainupdates{\cp}$ is the stream update set, a list of (stream identifier, new state, new bit flag) tuples: $\chainupdates{\cp} = [(\streamid_1, \streamstate'_{\streamid_1}, \streamnewbit_{\streamid_1}), \ldots, (\streamid_k, \streamstate'_{\streamid_k}, \streamnewbit_{\streamid_k})]$. The new bit flag indicates whether the identifier is new or already exists; the updated stream~$\stream'$ follows directly from $\stream$ and $\chainupdates{\cp}$. The size of the set~$\chainupdates{\cp}$ is a key metric when evaluating the performance of the \sysname{} protocols, and we use the term updates/s (updates per second) to refer to the number of stream updates produced or processed in a second.

\begin{definition}[Stream]\label{def:stream}
  Given an identifier type, value type, and a stream update
  function~$\StreamUpdateFn$, a stream~$\stream$ is a collection of
  stream identifiers and
  their corresponding states: $\stream = \{ (\streamid_1,
    \streamstate_{\streamid_1}), \ldots, (\streamid_M,
    \streamstate_{\streamid_M}) \}$.
\end{definition}

The stream abstraction is general and can represent a variety of indexing schemes. An address-balance stream that tracks the balance of each address uses the address as its identifier (e.g., $\streamid = {\tt 0x123}$) and the current balance as its state. Its stream update function extracts all balance updates for this address $\{ \streamupdate_1, \ldots, \streamupdate_n \}$ from the checkpoint contents and updates the balance as $\streamstate'_{\streamid} \gets \streamstate_{\streamid} + \sum_{i=1}^n \streamupdate_i$. An event stream tracking all events of a given type uses the event type as its identifier (e.g., $\streamid = {\tt exampleContract::exampleEvent}$) and a hash-chain commitment of all matching events as its state, updating by extracting matching events $\{ \streamupdate_1, \ldots, \streamupdate_n \}$ and folding them into the hash-chain as $\streamstate'_{\streamid} \gets H(\ldots H(H(\streamstate_{\streamid}, \streamupdate_1), \streamupdate_2) \ldots, \streamupdate_n)$. A blockchain may support object, address-object, or custom developer-specified streams that define a stream update function over filtered event types or the union of multiple event types.

\subsection{Goals} \label{sec:goals}

Our goal is to support verifiable queries for the state of any stream at the end of any checkpoint~$\cp$, a property we call \emph{completeness}. Let $M$ denote the total number of stream IDs and $k$ the number of stream IDs updated per checkpoint. Our solution should fulfill two requirements. \textbf{R1} (Completeness for varied streams): Checkpoints must contain authenticated data structures that permit querying the state of any of several streams at an arbitrary checkpoint. \textbf{R2} (Light-weight validator changes): The bandwidth, memory, and computational resources a validator spends to build a checkpoint are linear in the size of the checkpoint (i.e., $O(k)$) and in particular do not depend on $M$. We do not place any restrictions on the internal storage of validators, so it is acceptable if the internal storage increases by $O(M)$ to support a new stream.

\subsection{API and Correctness}
\label{sec:api}

We next define the key APIs used by the three participants in \sysname{}: validators~({\bf Val}), the ZK service~({\bf ZKS}), and clients~({\bf Client}). Validators produce a \emph{stream update commitment}~$\localdigest{\cp}$ for each checkpoint, which compactly commits to all stream updates~$\chainupdates{\cp}$ derived from that checkpoint. The ZK service maintains a commitment over the full stream~$\stream$ (that consists of all stream identifiers and their corresponding states) called the \textit{global stream commitment}~$\globaldigest{\cp}$. Clients are interested in learning the state of a stream identifier~$\streamid$ at a specific checkpoint~$\cp$.

\begin{enumerate}
  \item {\bf [Val]} $\localdigest{\cp} \gets \localcommit(\chainupdates{\cp})$: Executed by validators before committing to a checkpoint, this commits to the stream update set $\chainupdates{\cp}$ of the $k$ streams updated in checkpoint~$\cp$, optionally taking prior digests as inputs.
  \item {\bf [Val]} $\header{\cp} \gets \genheader(\cp, \localdigest{\cp}, \ldots)$: The checkpoint header generating function is modified to include $\localdigest{\cp}$ as a new field.
  \item {\bf [ZKS]} $(\globaldigest{\cp}, \updateproof{\cp}) \gets \processcp(\globaldigest{\cp-1}, (\header{\cp}, \chainupdates{\cp}))$: The ZK service processes the new checkpoint updates to advance the global commitment.
  \item {\bf [ZKS]} $(\streamstate, \validityproof) \gets {\sf readState}(\streamid, \cp)$: Returns the stream state after checkpoint $\cp$ for the stream identifier $\streamid$ along with a validity proof.
  \item {\bf [Client]} $0/1 \gets {\sf verifyState}(\streamid, \cp, \streamstate, \validityproof)$: Verifies the validity proof assuming knowledge of the public keys of the validator quorum.
\end{enumerate}

Intuitively, a stream state $(\streamid, \streamstate)$ is \emph{correct} at checkpoint~$c$ if it incorporates all updates to $\streamid$ up to~$c$ and excludes any that occur afterward.

\begin{definition}[Stream State Correctness]\label{def:stream-correctness}
  A stream state $(\streamid, \streamstate)$ is \emph{correct} at
  checkpoint~$c$ if there exists a checkpoint~$c' \le c$ and a sequence
  of stream updates
  $\chainupdates{c'}, \chainupdates{c'+1}, \ldots, \chainupdates{c}$ such that:
  \begin{enumerate}
    \item For every $c'' \in [c', c]$, the stream update commitment
          $\localdigest{c''} = \localcommit(\chainupdates{c''})$ is included in a
          valid checkpoint header $\header{c''}$, i.e., ${\sf
              IsValid}(\header{c''}) = 1$.
    \item The stream $\streamid$ is updated at checkpoint~$c'$,
          i.e., $(\streamid, \streamstate, \_) \in \chainupdates{c'}$,
          and remains unchanged thereafter: $\forall\, c'' \in (c', c],\;
            (\streamid, \_) \notin \chainupdates{c''}$.
  \end{enumerate}
\end{definition}

\paragraph{Merkle trees.}
The \sysname{} protocols use three standard Merkle-tree operations:
${\sf MT.vf}$ verifies a leaf against a root, ${\sf MT.vfAndUpd}$
additionally updates the leaf and returns the new root, and
${\sf MT.insert}$ appends a leaf at a given index
(specified in \Cref{app:merkle}).

\section{The \sysname{} protocols}

The ZK service maintains a verifiable commitment to all stream states $\{(\streamid_i, \streamstate_i)\}_{i=1}^{M}$ by processing each checkpoint within a ZK circuit, enabling efficient inclusion proofs for clients to query any stream identifier at any checkpoint. Our main protocols employ a Merkle tree because it supports efficient proofs even when tracking all stream identifiers (e.g., $M \approx 2^{30}$), effectively acting as an \emph{authenticated full node}. We first describe validator commitments (\Cref{sec:validator-commitment}), then present \minnowMTA{} (\Cref{sec:minnow-mt-a}) for moderate throughput and \minnowMTB{} (\Cref{sec:minnow-mt-b}) for arbitrarily high throughput; \Cref{sec:minnow-h} describes a hash-based variant for a small fixed set of streams.

\subsection{Validator commitment}
\label{sec:validator-commitment}

Validators commit to the states of all streams created or updated in the current checkpoint by partitioning the update set into fixed-size batches and hashing them into a hash chain. Recall from \Cref{sec:model} that each checkpoint $\contents{\cp}$ yields a stream update set $\chainupdates{\cp}$ of size $\numchainupdates{\cp}$, consisting of updates $\streampoint_{i} = (\streamid_i, \streamstate_i, \streamnewbit_i)$. Validators first partition $\chainbundle{\cp}$ into $\numcircuitbatches = \lceil \numchainupdates{\cp} / \batchsize \rceil$ batches $\{\batch{1}, \batch{2}, \ldots, \batch{n}\}$ of size $\batchsize$, padding the final batch if necessary. They then compute a hash $h_i = H(\batch{i})$ for each batch in parallel, and finally compute the stream update commitment as a hash chain over these batch hashes, where the previous checkpoint's commitment $\localdigest{\cp-1}$ initiates the chain:
\[
  \localdigest{\cp} = H(\ldots H(H(\localdigest{\cp-1}, h_1), h_2) \ldots, h_{\numcircuitbatches}).
\]
We denote this procedure compactly as $\localdigest{\cp} = \minnowhash(\localdigest{\cp-1}, \chainupdates{\cp})$, which implements the $\localcommit$ function from the API in \Cref{sec:api}. The resulting chain head $\localdigest{\cp}$ is included in the checkpoint header $\header{\cp}$ to fulfill the $\genheader$ function. Batching enables parallel commitment generation and later allows parallel proof generation in the ZK pipeline.

\subsection{A first attempt: \minnowMTA{}}
\label{sec:minnow-mt-a}

In \minnowMTA{}, the ZK service processes stream updates in fixed-size bundles of $\numcircuitupdates = \numcircuitbatches \cdot \batchsize$ updates, processing multiple bundles sequentially if a checkpoint exceeds this size. The service maintains a Merkle tree over all stream states, where each leaf stores the latest state $(\streamid, \streamstate)$ and empty leaves represent uninitialized streams. It updates the tree verifiably after processing each bundle. Assuming the circuit has processed all checkpoints up to $\cp$, its public inputs summarize the current system state using the latest Merkle root $\globaldigest{\cp}$, the number of non-empty leaves $\numleaves{\cp}$, the latest validator commitment $\localdigest{\cp}$, and the checkpoint number $\cp$. The circuit takes the most recent proof it has generated, the new bundle of stream updates $\chainbundle{\cp+1} = \{(\streamid_i, \streamstate_i, \streamnewbit_i)\}$ and their corresponding Merkle paths as inputs. It first verifies the previous proof, and then applies each stream update: if an existing stream is updated ($\streamnewbit = 0$) it verifies the Merkle path and updates the leaf with the new value, whereas if a new stream is introduced ($\streamnewbit = 1$) it checks that the insertion occurs at the next available leaf index indicated by $\numleaves{}$.

\paragraph{Circuit parameters and inputs.}
At setup time we fix and hard-code into the circuit the tree depth $\treedepth$ to support $\maxnumstreams = 2^{\treedepth}$ stream IDs, the validator batch size $\batchsize$, and the number of batches per bundle $\numcircuitbatches$ so one invocation processes $\numcircuitupdates = \numcircuitbatches \cdot \batchsize$ updates. Cheap recursion makes these parameters easy to change, since a new circuit simply verifies the old proof in its base case and continues from there. The private inputs are the previous proof with its public inputs along with the old leaf, the new leaf, and the Merkle path for every update, while the public inputs are $\pubinputs{\cp} = \{\globaldigest{\cp}, \localdigest{\cp}, \numleaves{\cp}, \cp\}$. \Cref{app:circuit-a} specifies the full statement $\minnowMTAstatement$.

\paragraph{Serving proofs to clients.}
To answer ${\sf readState}(\streamid, \cp)$, the ZK service returns the state together with $\validityproof = \{\updateproof{\cp}, \header{\cp}, \merklepath_{\streamid}\}$, which consists of the recursive proof with its public inputs, the signed checkpoint header carrying $\localdigest{\cp}$, and the Merkle path of the queried stream. The ${\sf verifyState}$ routine checks the Merkle path against $\globaldigest{\cp}$, verifies the validator signature on the header and checks that it carries $\localdigest{\cp}$ and $\cp$, and finally checks the proof $\updateproof{\cp}$ (\Cref{app:client}).

\paragraph{Security theorem.}
We now state and prove the security of \minnowMTA{}.

\begin{theorem}[\minnowMTA{} Security]
  \label{thm:minnow-security}
  Let $B$ be a blockchain as modeled in \Cref{sec:model}: a sequence of
  valid checkpoints whose headers are signed by the committee, with
  committee handoffs authenticated from a public genesis. Suppose a client
  holding the committee key~$\validatorpk$ runs ${\sf verifyState}(\streamid,
    \cp, \streamstate, \pi)$ against $B$ as specified above. If the verifier
  outputs $1$, and if the following assumptions hold:
  \begin{enumerate}
    \item The signature scheme is existentially unforgeable under
          chosen-message attack.
    \item The hash function used in $\minnowhash$ and Merkle trees is
          collision resistant.
    \item The employed ZK proof system is knowledge-sound for the
          statement~$\minnowMTAstatement$ (including recursive verification).
  \end{enumerate}
  then $(\streamid, \streamstate)$ is correct at checkpoint~$\cp$ of $B$,
  except with negligible probability.
\end{theorem}

Informally, the theorem holds because the hash-chain structure of the validator commitment ($\minnowhash$) ensures that verifying the signature on the final header guarantees the correctness of all prior stream update commitments; this relies on signature unforgeability and hash collision resistance. The knowledge-soundness of the recursive ZK proofs further guarantees that all updates to the Merkle root match the stream updates. \Cref{sec:minnow-a-proof} gives the formal proof under the stated assumptions.

\paragraph{Cost analysis.}
Let $h$ denote the constraint cost of a single hash invocation and $r$ the cost of verifying one ZK proof. Each circuit invocation performs the ZK verification costing $r$ constraints, Merkle operations requiring $2 \cdot \numcircuitupdates \cdot \treedepth$ hashes to verify and update each leaf, and stream update commitments needing $\numcircuitupdates + \numcircuitupdates / \batchsize$ hashes to compute $\minnowhash$, bringing the total cost to approximately $r + \numcircuitupdates \cdot h \cdot (2\treedepth + 1 + 1/\batchsize)$ constraints. A key optimization in \minnowMTA{} avoids verifying validator signatures inside the ZK circuit, which would otherwise be prohibitively expensive or require validators to adopt a ZK-friendly signature scheme.

\ifextended
  Rather than proving each checkpoint individually, multiple checkpoints
  can be processed together to amortize the fixed recursion cost~$r$.
  However, this sacrifices the ability to query intermediate checkpoints.
  Whether this is an issue depends on the use case and rate of
  checkpoint production.
  For example, the Sui blockchain produces four checkpoints per second,
  so we believe
  batching could be a reasonable approach.
  We explore the performance benefit of batching concretely in \Cref{sec:eval}.

  Another benefit of batching is that it allows handling dynamic
  variations in throughput.
  For example, if \minnowMTA{}'s circuits can handle $N$ updates, and
  if two consecutive checkpoints have $1.5N$ and $0.5N$ updates,
  then processing the two checkpoints together means that the ZK service
  just needs to process two bundles instead of three.
\fi

\paragraph{Limitations.}
The main limitation of \minnowMTA{} is its sequential dependency, as each proof generation depends on the previous proof as an input and must complete before the next bundle of updates arrives. If proving a bundle takes $t^A$ and the blockchain produces a bundle every $\bundletime$ seconds on average, the system requires $t^A \leq \bundletime$; otherwise latency grows rapidly and the system stalls (\Cref{fig:minnow-protocols}). Our evaluation in \Cref{sec:eval} shows that \minnowMTA{} can keep up with roughly 290 updates per second for a Merkle tree of size $2^{30}$.

\subsection{Our main protocol: \minnowMTB{}}
\label{sec:minnow-mt-b}

\begin{figure}[t]
    \centering
\begin{tikzpicture}[
    node distance=1cm and 1.2cm,
    >=Stealth,
    prover/.style={
        ellipse, draw, minimum width=1.3cm, minimum height=0.7cm, inner sep=1pt
    },
    every node/.style={
            font=\scriptsize
    }]

\foreach \x in {0,...,3} {
    \node[rectangle, draw, minimum width=1.05cm, minimum height=0.6cm, xshift=1.15*\x cm] (SubA\x) {SB \x};
}

\draw[dotted, thick] ([yshift=0.15cm,xshift=-0.12cm]SubA0.north west) rectangle ([yshift=-0.15cm,xshift=0.12cm]SubA3.south east);
\node[above=0.1cm of SubA1.north east, anchor=south] {Bundle $i$};

\foreach \x in {0,...,3} {
    \node[rectangle, draw, minimum width=1.05cm, minimum height=0.6cm, right=1.2cm of SubA3, xshift=1.15*\x cm] (SubB\x) {SB \the\numexpr4+\x\relax};
}

\draw[dotted, thick] ([yshift=0.15cm,xshift=-0.12cm]SubB0.north west) rectangle ([yshift=-0.15cm,xshift=0.12cm]SubB3.south east);
\node[above=0.1cm of SubB1.north east, anchor=south] {Bundle $i+1$};

\foreach \x in {0,...,3} {
    \node[ellipse, draw, minimum width=0.95cm, minimum height=0.6cm, inner sep=1pt, below=0.55cm of SubA\x] (ProveA\x) {$\parallelphasestatement$};
    \draw[->] (SubA\x.south) -- (ProveA\x.north);
}

\foreach \x in {0,...,3} {
    \node[ellipse, draw, minimum width=0.95cm, minimum height=0.6cm, inner sep=1pt, below=0.55cm of SubB\x] (ProveB\x) {$\parallelphasestatement$};
    \draw[->] (SubB\x.south) -- (ProveB\x.north);
}

\node[prover] (AggA0) at ($(ProveA0)!0.5!(ProveA1) + (0,-1.35)$) {$\intermediatephasestatement^0$};
\node[prover] (AggA1) at ($(ProveA2)!0.5!(ProveA3) + (0,-1.35)$) {$\intermediatephasestatement^0$};

\draw[->] (ProveA0.south) -- node[left] {$\parallelphaseproof{0}$}  (AggA0.north);
\draw[->] (ProveA1.south) -- node[right] {$\parallelphaseproof{1}$} (AggA0.north);
\draw[->] (ProveA2.south) -- node[left] {$\parallelphaseproof{2}$}  (AggA1.north);
\draw[->] (ProveA3.south) -- node[right] {$\parallelphaseproof{3}$} (AggA1.north);

\node[prover] (AggB1) at ($(ProveB2)!0.5!(ProveB3) + (0,-1.35)$) {$\intermediatephasestatement^0$};
\node[prover] (AggB0) at ($(ProveB0)!0.5!(ProveB1) + (0,-1.35)$) {$\intermediatephasestatement^0$};

\draw[->] (ProveB0.south) --  node[left] {$\parallelphaseproof{4}$}   (AggB0.north);
\draw[->] (ProveB1.south) --  node[right] {$\parallelphaseproof{5}$}   (AggB0.north);
\draw[->] (ProveB2.south) --  node[left] {$\parallelphaseproof{6}$}   (AggB1.north);
\draw[->] (ProveB3.south) --  node[right] {$\parallelphaseproof{7}$}   (AggB1.north);

\node[prover] (FinalAggA) at ($(AggA0)!0.5!(AggA1) + (0,-1.35)$) {$\intermediatephasestatement^1$};
\draw[->] (AggA0.south) -- node[left] {$\intermediatephaseproof{0}{1}$}  (FinalAggA.north);
\draw[->] (AggA1.south) -- node[right] {$\intermediatephaseproof{2}{3}$} (FinalAggA.north);

\node[prover] (FinalAggB)  at ($(AggB0)!0.5!(AggB1) + (0,-1.35)$) {$\intermediatephasestatement^1$};
\draw[->] (AggB0.south) -- node[left] {$\intermediatephaseproof{4}{5}$} (FinalAggB.north);
\draw[->] (AggB1.south) -- node[right] {$\intermediatephaseproof{6}{7}$} (FinalAggB.north);

\node[prover, below=0.55cm of FinalAggA] (SeqA) {$\sequentialphasestatement$};
\node[prover, below=0.55cm of FinalAggB] (SeqB) {$\sequentialphasestatement$};

\draw[->] (FinalAggA.south) -- node[right] {$\intermediatephaseproof{0}{3}$} (SeqA.north);
\draw[->] (SeqA.east) -- node[above]{$\sequentialphaseproof{i}$} (SeqB.west);
\draw[->] (FinalAggB.south) -- node[right] {$\intermediatephaseproof{4}{7}$} (SeqB.north);
\draw[->] (SeqB.east) -- node[above]{$\sequentialphaseproof{i+1}$} ++(1.4cm,0);

\end{tikzpicture}
\caption{\minnowMTB{} with $\numsubbundles{}=4$ base proofs and an aggregation factor of $\aggregationfactor{}=2$. Processing a bundle happens in three phases: parallel ($\parallelphasestatement$), aggregation ($\intermediatephasestatement$), and sequential ($\sequentialphasestatement$).}
    \label{fig:minnow-b}
\end{figure}
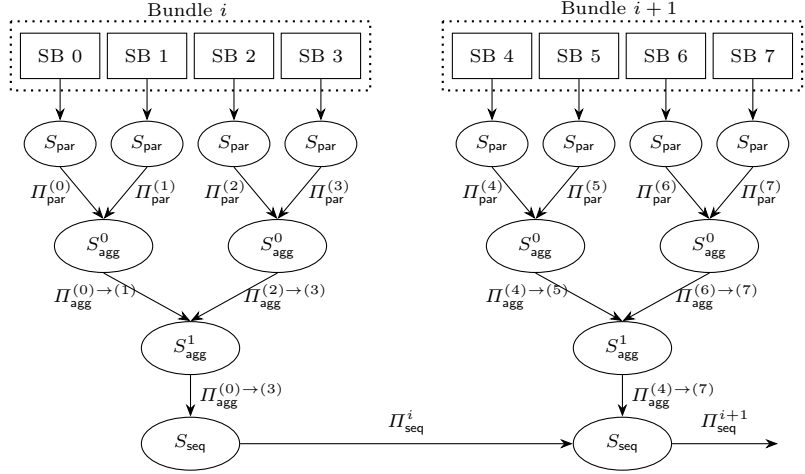

We now present \minnowMTB{}, an improved protocol that uses the same validator commitment function $\minnowhash$ as \minnowMTA{} but splits its monolithic statement into smaller components that can be proven in parallel to sustain arbitrarily high throughput. Every invocation of \minnowMTB{} processes a bundle, which may correspond to updates from one or several checkpoints and is subdivided into $\numsubbundles{}$ sub-bundles that act as the units of parallel processing. Each sub-bundle contains $\numbatchespersubbundle$ batches of $\batchsize$ updates each, matching the batch size validators use. Assuming for simplicity that all updates produced by checkpoint $\cp$ fit into one bundle where $\numchainupdates{\cp} \le \bundlesize$, \minnowMTB{} executes in three phases depicted in \Cref{fig:minnow-b}. The parallel phase first processes sub-bundles independently with a base circuit to update the Merkle root from $\globaldigest{\cp-1}$ to $\globaldigest{\cp}$, yielding $\numsubbundles{}$ base proofs. The aggregation phase then aggregates these base proofs into a single proof using a tree of aggregation circuits. Finally, the sequential phase uses a cyclic circuit to verify its own previous proof together with the aggregated proof for checkpoint $\cp$, producing one final proof per checkpoint.

\paragraph{Parameters.}
The protocol defines the tree depth $\treedepth$, the batch size $\batchsize$ representing updates per batch, and the sub-bundle size $\subbundlesize$ where $\subbundlesize \bmod \batchsize = 0$ so the number of batches per sub-bundle is $\numbatchespersubbundle = \subbundlesize/\batchsize$. It also defines the bundle size $\bundlesize$ where $\bundlesize \bmod \subbundlesize = 0$ so the number of sub-bundles per bundle is $\numsubbundles{} = \bundlesize / \subbundlesize$, and the aggregation factor $\aggregationfactor$ which determines that the proof tree has $\aggregationlevels = \log_{\aggregationfactor}(\numsubbundles{})$ levels.

\paragraph{1) Parallel phase.}
This phase chunks the bundle into $\numsubbundles{}$ sub-bundles and runs the base prover on each in parallel to prove the statement $\parallelphasestatement$, asserting that the sub-bundle transforms the state $\{\globaldigest{\sf old}, \localdigest{\sf old}, \numleaves{\sf old}\}$ into $\{\globaldigest{\sf new}, \localdigest{\sf new}, \numleaves{\sf new}\}$. The circuit closely mirrors \minnowMTA{}, with the key differences that it does not recursively verify a prior proof and it exposes both the old and new states as public inputs. Assuming sufficient machines, this phase outputs $\numsubbundles{}$ base proofs $\parallelphaseproof{i}$ while completing in the time required to generate one base proof.

\paragraph{2) Aggregation phase.}
The aggregation phase combines the $\numsubbundles{}$ base proofs into a single proof using an aggregation tree. Each level-$i$ aggregation circuit takes $\aggregationfactor$ proofs from the level below (base proofs at level 0), verifies them, checks that the outputs of one match the inputs of the next, and emits a combined proof. The public inputs of each aggregation circuit match those of $\parallelphasestatement$. For example, if $\numsubbundles{} = 4$ and $\aggregationfactor = 2$, the system uses two level-0 aggregations that each combine two base proofs, followed by one level-1 aggregation combining those two results as shown in \Cref{fig:minnow-b}. This phase is fully parallelizable, requiring $\numsubbundles{}/\aggregationfactor^{i+1}$ machines at level $i$.

\paragraph{3) Sequential phase.}
The cyclic circuit ties the system together by taking the last sequential proof $\sequentialphaseproof{\cp-1}$ for checkpoint $\cp-1$ and the aggregated proof $\intermediatephaseproof{\cp-1}{\cp}$ for the current checkpoint bundle, verifying both proofs, and matching their public inputs. Its statement $\sequentialphasestatement$, detailed in \Cref{app:cyclic}, outputs the next sequential proof $\sequentialphaseproof{\cp}$ to serve as the final proof for checkpoint $\cp$. The bottom half of \Cref{fig:minnow-protocols} illustrates the special case where $\numsubbundles{} = 1$ and $\aggregationfactor = 1$, whereas \Cref{fig:minnow-b} shows a more complex instance with $\numsubbundles{} = 4$ and $\aggregationfactor = 2$.

\paragraph{Several checkpoints per bundle.}
While the description above assumed \minnowMTB{} processes a single checkpoint per run, the system can process multiple checkpoints at once by choosing a fixed time slot $\bundletime$ and treating all checkpoints occurring in that slot as one bundle. As long as the updates produced in a slot remain below $\bundlesize$, \minnowMTB{} operates normally and only requires that the constant runtime of a single sequential phase invocation remains smaller than $\bundletime$. Because we control the slot length and the sequential phase takes roughly $0.63$\,s as shown in \Cref{sec:eval}, any choice of $\bundletime \ge 0.63$\,s avoids throughput caps. Fixing the batch size $\batchsize$, sub-bundle size $\subbundlesize$, and aggregation factor $\aggregationfactor$ to support a target throughput $\throughput$ updates per second leaves only the bundle size $\bundlesize$, which we set to $\bundlesize = \subbundlesize \cdot \lceil \throughput \cdot \bundletime / \subbundlesize \rceil$ so that it is a multiple of $\subbundlesize$.

\paragraph{Latency.}
Let $\parallelphasetime$, $\intermediatephasetime$, and $\sequentialphasetime$ denote the runtimes of the parallel, aggregation, and sequential phases respectively. The parallel phase $\parallelphasetime = O(\subbundlesize)$ behaves as $O(1)$ since all sub-bundles are proved in parallel, the aggregation phase $\intermediatephasetime = O(\aggregationlevels)$ performs $O(1)$ work across $\aggregationlevels = \log_{\aggregationfactor}(\numsubbundles{})$ levels, and the sequential phase $\sequentialphasetime = O(1)$. Thus the end-to-end latency is
\[
  \latency
  = \parallelphasetime
  + \intermediatephasetime
  + \sequentialphasetime
  = O(\log_{\aggregationfactor}(\numsubbundles{})).
\]
Since $\numsubbundles{} = \lceil \bundlesize / \subbundlesize \rceil = \lceil \throughput \cdot \bundletime / \subbundlesize \rceil$, the overall latency is $\latency = O(\log_{\aggregationfactor}(\throughput))$, logarithmic in throughput. The number of machines \minnowMTB{} needs to keep pace with the chain follows from the three phase runtimes (\Cref{app:machines}).

\section{Implementation and Evaluation}
\label{sec:eval}

We evaluate the \sysname{} protocols using a Rust implementation. It includes validator-side commitment logic, Merkle-tree maintenance, and ZK circuits for \minnowMTA{} and \minnowMTB{} built on the Plonky2 proving system~\cite{plonky2} with its default recursion configuration (codeword rate $1/8$, no zero-knowledge), as well as orchestration and networking code for end-to-end evaluation on a multi-machine testbed. In total, we have written 10.5k lines of code (4k for circuits/logic, 6.5k for orchestration), and we are open-sourcing the entire codebase to enable full reproducibility of our experiments.\footnote{\codelink} All benchmarks use AWS \texttt{m6a.8xlarge} instances within a single data center, where each instance provides 12.5\,Gbps network bandwidth, 32~vCPUs (16 physical cores) powered by a 3.6\,GHz AMD EPYC~7R13 processor, 128\,GB RAM, and Ubuntu~22.04.

\newcommand{\provingtime}[1]{t^{\sf prove}_{#1}}

\paragraph{Plonky2 prover.}
We begin by benchmarking the Plonky2 prover on circuits with increasing gate counts. A circuit with $g$ gates has polynomial degree $d =\lceil \log_2(g) \rceil$, and the corresponding proving time $\provingtime{d}$ is shown in \Cref{tab:proving-times}; for example, $\provingtime{14} = 0.627$\,s.

\begin{table}[t]
  \begin{minipage}[t]{0.54\textwidth}
    \centering
    \small
    \begin{tabular}{lrrrrr}
      \toprule
      Degree $d$                         & 13   & 14   & 15   & 16   & 17   \\
      \midrule
      Proving time $\provingtime{d}$ (s) & 0.32 & 0.63 & 1.29 & 2.77 & 5.96 \\
      \bottomrule
    \end{tabular}
    \caption{Plonky2 proving time by circuit degree ($g$ gates give degree $\lceil \log_2 g \rceil$).}
    \label{tab:proving-times}
  \end{minipage}\hfill
  \begin{minipage}[t]{0.43\textwidth}
    \centering
    \small
    \begin{tabular}{lrrrr}
      \toprule
                                 & \multicolumn{4}{c}{Degree}                     \\
      \cmidrule(lr){2-5}
      Streams ($\maxnumstreams$) & 14                         & 15  & 16   & 17   \\
      \midrule
      $2^{10}$                   & 300                        & 850 & 1,950 & 4,100 \\
      $2^{20}$                   & 200                        & 500 & 1,150 & 2,450 \\
      $2^{30}$                   & 100                        & 350 & 800  & 1,750 \\
      \bottomrule
    \end{tabular}
    \caption{Max updates per \minnowMTA{} proof at Plonky2 degrees 14--17 (batch size 50).}
    \label{tab:num-updates-A}
  \end{minipage}
\end{table}

\paragraph{Data collection.}
We analyze one month of Sui blockchain activity (May 25 to June 25, 2025) to estimate the throughput \sysname{} must sustain. We measure the average number of updates per second, our baseline, and the maximum number of updates in a single checkpoint, which multiplied by Sui's average checkpoint rate (4.12 checkpoints/s) gives a worst-case update rate. The events stream averages 86.18 events emitted per second with a worst-case $1{,}107$ events in a single checkpoint, corresponding to $1{,}107 \times 4.12 = 4{,}560.84$\,updates/s (peak observed at checkpoint~146264685). The objects stream averages 336.09 object updates per second with a worst-case $2{,}114$ objects updated in a single checkpoint, corresponding to $2{,}114 \times 4.12 = 8{,}709.68$\,updates/s (peak observed at checkpoint~144662971). These measurements suggest that \sysname{} must handle several thousand stream updates per second to support modern high-throughput chains.

\subsection{Validator overhead}
\label{sec:validator-overhead}
We measure the overhead that \sysname{} introduces on validators. Validators compute the \sysname{} commitment each checkpoint by chunking the $\numchainupdates{\cp}$ stream updates into batches of size~$\batchsize$, hashing each batch in parallel, and folding these batch hashes sequentially into the running commitment $\localdigest{\cp}$ (\Cref{sec:validator-commitment}). Only this folding step is sequential, so the batch size~$\batchsize$ controls the degree of parallelism. \Cref{fig:eval-micro}(a) shows the end-to-end time to compute $\localdigest{\cp}$ for $\numchainupdates{\cp} \in \{2{,}048, 4{,}096, 8{,}192\}$ using four CPU cores. A batch size of $\batchsize = 50$ strikes a good balance, keeping the overhead below 10\,ms and ensuring a negligible validator-side cost for \sysname{}. Since the largest checkpoint in our one-month Sui measurement carried 2,114 updates, a quarter of the 8,192 measured here, the per-checkpoint commitment cost is a small fraction of Sui's 240\,ms checkpoint interval. We thus fix $\batchsize = 50$ for all subsequent experiments, noting that the batch size barely affects circuit size: for a base circuit processing 1,000 updates, batch sizes of 10, 100, and 1,000 give 50,000, 49,890, and 49,876 gates (\Cref{tab:batch-size-experiments}).

\begin{figure}[t]
  \centering
  \includegraphics[width=\textwidth]{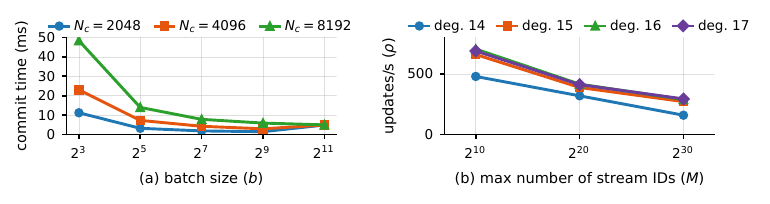}
  \caption{(a) Time taken to compute the \sysname{} commitment using 4~cores
    for varying batch sizes and update counts.
    (b) Amortized throughput of \minnowMTA{} for different maximum stream
    counts $\maxnumstreams$. Each curve corresponds to a different proving
    interval, equal to the Plonky2 proving time at circuit degrees~14--17
    (see \Cref{tab:proving-times}).}
  \label{fig:eval-micro}
\end{figure}

\subsection{\minnowMTA{}}
\label{sec:minnow-a-eval}
We evaluate the performance of \minnowMTA{}, where each circuit invocation processes a fixed number of stream updates~$\numcircuitupdates$ and performs one recursive proof verification. For a given maximum number of stream identifiers~$\maxnumstreams$ (which fixes the Merkle-tree depth $\treedepth = \log_2(\maxnumstreams)$), increasing $\numcircuitupdates$ increases the circuit size and thus its Plonky2 proving time. We consider three choices for~$\maxnumstreams$, namely $\maxnumstreams = 2^{10}, 2^{20}, 2^{30}$, and for each $\maxnumstreams$ and each Plonky2 degree $d \in \{14, 15, 16, 17\}$, we determine the largest $\numcircuitupdates$ such that the total number of gates still fits within the corresponding degree bound. We compile the circuit at increments of fifty updates and record the largest $\numcircuitupdates$ before the gate count crosses a power-of-two threshold (\Cref{tab:num-updates-A}).

Since one \minnowMTA{} proof is produced every proving interval $\provingtime{d}$ (for degree~$d$), the amortized throughput at that configuration is $\throughput \approx \frac{\numcircuitupdates}{\provingtime{d}}\;\text{updates/s}$. \Cref{fig:eval-micro}(b) plots this throughput for each $\maxnumstreams$ and degree $d \in \{14, 15, 16, 17\}$ using the measured Plonky2 proving times from \Cref{tab:proving-times}. When $\maxnumstreams = 2^{30}$, \minnowMTA{} can handle $100$ updates per proof yielding $\approx 159$\,updates/s at degree~14, $350$ updates per proof yielding $\approx 271$\,updates/s at degree~15, $800$ updates per proof yielding $\approx 289$\,updates/s at degree~16, and $1{,}750$ updates per proof yielding $\approx 294$\,updates/s at degree~17. Two opposing effects shape these curves: batching more updates into a single proof improves amortization of the recursive-verification cost (\Cref{sec:minnow-mt-a}), but Plonky2 proving time grows super-linearly with circuit size, meaning overly large circuits reduce throughput. The first effect is visible for $\maxnumstreams = 2^{30}$, where increasing the proving interval monotonically raises throughput. The second dominates for $\maxnumstreams \in \{2^{10}, 2^{20}\}$, where throughput peaks at an intermediate degree (16) and additional batching yields negative returns. For~$\maxnumstreams = 2^{30}$, the best performance is achieved by processing all the checkpoints arriving in a window of 5.96\,s at once, and \minnowMTA{} then processes about 294 updates per second; in practice one may prefer a shorter proving interval to answer queries at finer time granularity. These numbers motivate \minnowMTB{} for high-throughput chains emitting thousands of updates per second. Plonky2 proof sizes for \minnowMTA{} range between 120 and 150\,KB depending on the number of gates; at $\maxnumstreams = 2^{30}$ and $\numcircuitupdates = 800$ (degree~16), the proof is about 141.48\,KB. A full \sysname{} proof additionally includes the checkpoint header (roughly 300\,B on Sui) and a Merkle path, for a total of around 150\,KB.

\subsection{\minnowMTB{}}
\label{sec:minnow-b-eval}
\paragraph{Microbenchmarks.}
\label{sec:minnow-b-microbenchmarks}
We configure \minnowMTB{} with a fixed maximum number of streams~$\maxnumstreams = 2^{30}$. The parallel circuit~$\parallelphasestatement$ is structurally identical to the \minnowMTA{} circuit except that it omits recursive proof verification, reducing its gate count and allowing it to process more updates per invocation: 200, 450, 900, and 1,900 updates at degrees 14 to 17 (\Cref{tab:num-updates-par}). Each aggregation circuit verifies $\aggregationfactor$ proofs and consists of approximately $4{,}500 \cdot \aggregationfactor$ gates. The cyclic circuit verifies both the previous cyclic proof and the final aggregated proof; its roughly $10$k gates give it Plonky2 degree~14 and proving time $\sequentialphasetime = 0.63$\,s (\Cref{tab:proving-times}). The key design requirement is that the bundle interval exceeds this cyclic proving time, $\bundletime \ge \sequentialphasetime = 0.63$\,s.

\paragraph{Special configuration.}
For concreteness, we analyze a configuration with
$\bundletime = 0.63$\,s, $\batchsize = 50$, $\subbundlesize = 200$, and
$\aggregationfactor = 3$. With these choices, both the parallel and aggregation circuits fit comfortably within degree~14, so their proving times also equal~0.63\,s. If a bundle contains $\beta$ sub-bundles ($\beta = \bundlesize / y$), then the aggregation tree has \(\aggregationlevels = \log_3(\beta)\) levels. Thus, the end-to-end latency is $\latency = \parallelphasetime + \aggregationlevels \cdot \intermediatephasetime + \sequentialphasetime = 0.63 \cdot (2 + \aggregationlevels)$\,s. Bundles of up to 200, 600, 1,800, and 5,400 updates need 0, 1, 2, and 3 aggregation levels, giving predicted latencies of 1.26, 1.89, 2.52, and 3.15\,s at 317, 952, 2,857, and 8,571\,updates/s on 2, 5, 14, and 41 machines. Latency grows logarithmically in the bundle size while throughput grows linearly, the main scalability advantage of \minnowMTB{}; \Cref{tab:latency-throughput} tabulates these values and \Cref{app:machines} derives the machine counts.

\begin{figure}[t]
  \centering
  \includegraphics[width=\textwidth]{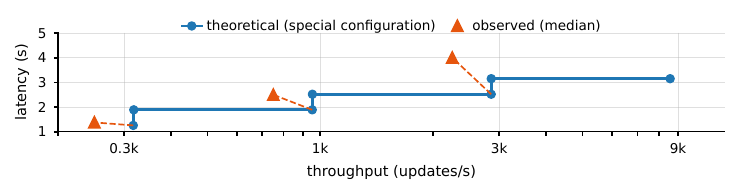}
  \caption{Theoretical (blue) and observed (red) median latency vs.\ throughput
    for \minnowMTB{}. The blue step curve corresponds to the special
    configuration
    in \Cref{tab:latency-throughput}; each plateau reflects a fixed aggregation
    level, and each step corresponds to adding a new aggregation level.}
  \label{fig:minnow-latency-throughput}
\end{figure}

\paragraph{End-to-end evaluation.}
\label{sec:end-to-end-evaluation}
We evaluate the end-to-end throughput and latency of \minnowMTB{} on a distributed testbed of up to 15 AWS instances (one primary, up to 13 workers, and one client) to assess its performance in a realistic deployment setting. For these experiments, we fix the bundle interval to $\bundletime = 0.8$\,s, slightly larger than the theoretical minimum of 0.63\,s from \Cref{sec:minnow-b-microbenchmarks}, to provide some breathing room for scheduling jitter and network variability. The implementation uses the tokio runtime~\cite{tokio} for asynchronous networking and scheduling CPU-intensive tasks, as well as raw TCP sockets for inter-machine communication. The system architecture consists of a single \emph{primary} machine and a configurable set of \emph{worker} machines, where the primary receives bundles of state updates, constructs the corresponding witnesses, and dispatches them to the workers. Each worker computes both base and aggregation proofs, and once all intermediate proofs are returned, the primary computes the final cyclic proof, returns it to the client, and orchestrates all client communications. Operators expose only the primary and scale by adding workers.

The red triangles in \Cref{fig:minnow-latency-throughput} show the observed latency and throughput of \minnowMTB{} for each bundle size. Each point is the median over a 3-minute run; latency runs from a client's submission to the final cyclic proof.
For a bundle size of 200 updates, the median latency is 1.4\,s (p90: 1.5\,s). Increasing the bundle size to 600 updates raises it to 2.5\,s (p90: 2.9\,s), and at 1,800 updates it reaches 4\,s (p90: 4.9\,s). The corresponding observed throughputs are 250, 750, and 2,250\,updates/s (one bundle per 0.8\,s). The number of machines used for each bundle size is exactly as predicted by the theoretical analysis, and the measurements confirm that the end-to-end latency of \minnowMTB{} grows sub-linearly with the bundle size as expected.

The remaining gap between theoretical and observed latency stems primarily from networking overhead, context switching on the primary, queuing delays during bundle preparation on clients, and the time required to construct witnesses. The per-prover breakdown in \Cref{fig:latency-breakdown} confirms this (\Cref{sec:breakdown} discusses it): the base, aggregation, and cyclic provers each take a median of about 0.63\,s regardless of the bundle size, while witness construction on the primary grows from 0.05\,s at 200 updates to 0.18\,s at 1{,}800. The small throughput gap between the theoretical and observed curves is fully explained by our choice of a slightly larger bundle interval, since we send a bundle every 0.8\,s in the experiments, whereas the idealized analysis earlier assumes $\bundletime = 0.63$\,s.

\begin{figure}[t]
  \centering
  \includegraphics[width=\textwidth]{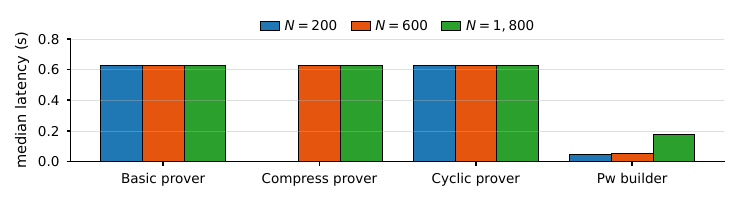}
  \caption{Latency breakdown of \minnowMTB{} by prover type, for bundle
    sizes of $N = 200$, $600$, and $1{,}800$ (1, 4, and 13 workers).}
  \label{fig:latency-breakdown}
\end{figure}

\section{Related work}

Prior light-client designs assume that validators publish, in every block header, a commitment to the full state; to our knowledge \sysname{} is the first to give light clients completeness on chains whose validators commit only to the updates of each checkpoint.

\paragraph{Light clients and state proofs.}
Header-chain compressions such as NIPoPoWs~\cite{nipopow}, FlyClient~\cite{flyclient}, and Blink~\cite{blink} for proof of work, and PoPoS~\cite{popos} for proof of stake, reduce downloads. Plumo~\cite{plumo} and zkBridge~\cite{zkbridge} go further with SNARK proofs of consensus, the latter proving sync-committee BLS signatures inside a distributed SNARK. Sync-committee clients include Altair~\cite{altair-light-client} and Helios~\cite{helios}, and Chatzigiannis et al.~\cite{chatzigiannis2022sok} survey the space. All of these verify headers and then read state through the state root in the header, which Sui and Solana do not publish. Other designs relax trust in different ways: fraud and data-availability proofs~\cite{al2018fraud} and light clients for lazy blockchains~\cite{lazyclient} need an honest full node, and BITE~\cite{matetic2019bite} and ZLiTE~\cite{wuest2019zlite} rely on trusted hardware with known limitations~\cite{nilsson2020survey}. \sysname{} trusts only the validator set. Mina~\cite{bonneau2020mina,cryptoeprint:2020/1522} makes validators themselves run incrementally verifiable computation, so proving is on-chain and heavy; \sysname{} keeps validator overhead to one header field. Sunfish~\cite{sunfish} is the closest work and the source of the validator-side hash-chain commitment. Its sparse client downloads and re-executes only the transactions touching a chosen sub-state and receives a succinct completeness proof that none was omitted. The value of a sub-state is thus obtained by re-execution, the cost grows with that sub-state's traffic, and there is no succinct inclusion proof for an arbitrary object at an arbitrary checkpoint. \sysname{} adds the untrusted ZK service that turns the same commitment into full-state, per-checkpoint inclusion proofs.

\paragraph{State commitments and stateless clients.}
Chains paying for a per-block state root use Merkle-Patricia tries (Ethereum~\cite{ethereum}) or Jellyfish Merkle trees~\cite{jellyfish} (Aptos~\cite{aptos}); LVMT~\cite{lvmt} shows this authenticated storage is the execution bottleneck. Verkle trees~\cite{verkle-buterin}, aggregatable subvector commitments~\cite{asvc}, and Hyperproofs~\cite{hyperproofs} shrink proofs or make them maintainable, but the commitment stays on the validators' critical path. High-throughput chains drop it: Sui checkpoints carry transaction and effects digests but no state root, and the committee commits to the live object set only at epoch end~\cite{sui-checkpoint-verification,sui-snapshots}. Solana commits per block only to the accounts written in that block, and its lattice-hash successor supports no inclusion proofs~\cite{solana-simd-0215}. That is the gap \sysname{} fills, without putting the commitment back on validators. Reckle trees~\cite{PSG+24} store recursive proofs at Merkle nodes to obtain succinct, updatable batch-membership and map-reduce proofs against a given root; they assume a trusted full-state digest such as Ethereum's header root and do not produce one. \sysname{} is complementary: its service could maintain a Reckle tree over the same leaves, adding succinct batched reads for standing queries on top of the per-checkpoint root proof.

\paragraph{ZK co-processors and verifiable indexers.}
Co-processors such as Axiom~\cite{axiom}, Brevis~\cite{brevis}, and Lagrange~\cite{lagrange-coprocessor} prove reads of historical state and computations over it, and zkVMs such as SP1~\cite{sp1} and RISC Zero~\cite{risc0} let developers build such indexers generically; validity rollups prove the execution of every transaction and post a state root~\cite{zk-rollups-bench}. All of them consume a state root that a header already provides. \sysname{} produces that root for chains that lack one, and proves only the tree update from validator-signed effects rather than re-executing transactions.

\paragraph{Recursion and distributed proving.}
Incrementally verifiable computation~\cite{valiant2008incrementally} and proof-carrying data~\cite{bitansky2013} underlie \sysname{}; recursion was long expensive~\cite{recursion-cycles}, and Halo~\cite{halo} and Plonky2~\cite{plonky2} made it practical. We use Plonky2 for its FRI-based, small-field recursion and its available implementation. Folding schemes such as Nova~\cite{nova}, HyperNova~\cite{hypernova}, and Protostar~\cite{protostar} are an alternative when a proof at every checkpoint is not required; we leave exploring them to future work. Mangrove~\cite{mangrove} obtains tree-shaped proof-carrying data with parallel proving in the folding setting, and zkTree~\cite{zktree} builds a recursion tree of Plonky2 proofs; both share the shape of \minnowMTB{}'s aggregation tree. Distributed provers such as DIZK~\cite{dizk}, Pianist~\cite{pianist}, and Hekaton~\cite{hekaton}, and aggregation schemes such as SnarkPack~\cite{snarkpack}, scale one large statement across machines. \minnowMTB{} instead exploits that its workload is naturally a tree of independent Merkle-tree updates, so it needs only small circuits and off-the-shelf recursion.

\ifextended
    
\section{Conclusion}

We introduced \sysname{}, a practical protocol for supporting
efficient light clients with completeness on high-throughput
blockchains. By shifting state maintenance to an off-chain ZK service and
requiring validators to commit only to per-checkpoint state updates,
\sysname{} avoids the heavy costs of maintaining global state
commitments on-chain. Our design combines hash-chain commitments,
recursive proofs, and a parallelizable proving pipeline, culminating
in \minnowMTB{}, which achieves logarithmic latency growth and
sustains arbitrarily high throughput given sufficient parallelism.

Our prototype implementation using Plonky2 shows that \sysname{} can
process thousands of state updates per second with only a few seconds
of end-to-end latency, while imposing negligible overhead on
validators. This demonstrates that completeness for light clients is
compatible with the performance requirements of modern blockchains.

Looking ahead, \sysname{} opens several directions for future work, three
of which we discuss below; others include additional query types such as
non-inclusion proofs or per-stream proofs, and integrating state-of-the-art
ZKPs~\cite{plonky3} to further reduce latency.
More broadly, the techniques developed here may benefit other
incrementally-verifiable computations that require both high
throughput and low latency.

\paragraph{Handling varying input rates.}
Because a blockchain emits a varying rate of stream updates while circuits process a fixed rate, both \minnowMTA{} and \minnowMTB{} handle this mismatch by setting a maximum bundle size. If the chain's average throughput grows beyond this limit, the system transitions to a new set of circuits supporting a larger bound, using the new sequential circuit to verify the previous sequential proof as its base case for smooth migration. This strategy allows the system to scale gradually with blockchain usage rather than requiring a large parameter choice from the start.

Temporary spikes are easily absorbed: if one slot produces more than
$\bundlesize$ updates but the next slot produces fewer, the excess can
be split into multiple bundles while maintaining overall latency
through batching.

A drawback of fixed-size circuits is that the ZK service pays for
$\bundlesize$ updates per slot even if fewer updates occur.
One solution is to design a cyclic circuit capable of verifying aggregation
circuits of multiple bundle sizes (e.g.,
$\bundlesize \in \{5{,}000, 10{,}000, 15{,}000\}$), allowing the system to choose the
best parameter on a per-slot basis.
We leave it to future work to investigate the efficiency of this approach.

\paragraph{Tracking a fixed set of streams.}
A useful variant of \sysname{} arises when the service tracks only a small fixed set of stream identifiers, which reduces proving cost while still providing completeness for selected streams.

In this case, we replace the Merkle-tree commitment used by
\minnowMTA{} and \minnowMTB{} with a simple hash over the tracked
states. If the service tracks a fixed set
$\streamidset = \{\streamid_1, \ldots, \streamid_m\}$, then its global
commitment is
$  \globaldigest{\cp}
= H\bigl([(\streamid_i,\streamstate_i)]_{i=1}^m\bigr).
$
As before, the service processes the checkpoint updates to detect
changes to any tracked stream and recompute the commitment.
We explore this variant further in \Cref{sec:minnow-h}.

\paragraph{Reducing validator work further.}
Depending on the stream type, supporting \sysname{} may require validators to maintain an additional key-value map. Supporting object streams on Sui requires no modification since validators already store the mapping from object IDs to their latest values, whereas supporting event streams requires introducing a new key-value store and addressing operational concerns such as state synchronization for new validators.

A natural question is whether even this overhead can be shifted to the
ZK service. In principle, \sysname{} could be extended so that the ZK
service also maintains per-stream state (e.g., rolling hash chains for
event streams) in addition to maintaining global commitments. Whether
this is practical depends on the cost of updating these structures
inside the ZK circuit, and exploring this trade-off is an interesting
direction for future work.

\fi
\ifpublish
    \paragraph{Acknowledgments.}
This work is partially funded by Mysten Labs. We thank Andrey Chursin and Michael
Corey for help with ideating and data collection respectively.

\fi

\bibliographystyle{splncs04}
\bibliography{references}

@misc{solana,
    title        = {Solana: A new architecture for a high performance blockchain v0.8.13},
    author       = {Yakovenko, Anatoly},
    howpublished = {Whitepaper},
    url          = {https://solana.com/solana-whitepaper.pdf},
    year         = {2018}
}

@misc{aptos,
    title        = {The Aptos Blockchain: Safe, Scalable, and Upgradeable Web3 Infrastructure},
    author       = {Aptos Labs},
    howpublished = {Whitepaper},
    url          = {https://aptosfoundation.org/whitepaper/aptos-whitepaper_en.pdf},
    year         = {2022}
}

@inproceedings{chatzigiannis2022sok,
    title        = {{SoK}: Blockchain Light Clients},
    author       = {Chatzigiannis, Panagiotis and Baldimtsi, Foteini and Chalkias, Konstantinos},
    booktitle    = {International Conference on Financial Cryptography and Data Security},
    pages        = {615--641},
    year         = {2022},
    organization = {Springer}
}

@article{nilsson2020survey,
    title   = {A survey of published attacks on Intel SGX},
    author  = {Nilsson, Alexander and Bideh, Pegah Nikbakht and Brorsson, Joakim},
    journal = {arXiv preprint arXiv:2006.13598},
    year    = {2020}
}

@inproceedings{wuest2019zlite,
    title        = {Zlite: Lightweight clients for shielded zcash transactions using trusted execution},
    author       = {W{\"u}st, Karl and Matetic, Sinisa and Schneider, Moritz and Miers, Ian and Kostiainen, Kari and {\v{C}}apkun, Srdjan},
    booktitle    = {Financial Cryptography and Data Security: 23rd International Conference, FC 2019, Frigate Bay, St. Kitts and Nevis, February 18--22, 2019, Revised Selected Papers 23},
    pages        = {179--198},
    year         = {2019},
    organization = {Springer}
}

@inproceedings{matetic2019bite,
    title     = {{BITE}: Bitcoin lightweight client privacy using trusted execution},
    author    = {Matetic, Sinisa and W{\"u}st, Karl and Schneider, Moritz and Kostiainen, Kari and Karame, Ghassan and Capkun, Srdjan},
    booktitle = {28th USENIX Security Symposium (USENIX Security 19)},
    pages     = {783--800},
    year      = {2019}
}

@inproceedings{popos,
    author    = {Agrawal, Shresth and Neu, Joachim and Tas, Ertem Nusret and Zindros, Dionysis},
    title     = {{Proofs of Proof-Of-Stake with Sublinear Complexity}},
    booktitle = {5th Conference on Advances in Financial Technologies (AFT 2023)},
    year      = {2023},
    publisher = {Schloss Dagstuhl -- Leibniz-Zentrum f{\"u}r Informatik},
    url       = {https://drops.dagstuhl.de/entities/document/10.4230/LIPIcs.AFT.2023.14}
}

@inproceedings{nipopow,
    author    = {Aggelos Kiayias and Andrew Miller and Dionysis Zindros},
    title     = {Non-interactive Proofs of Proof-of-Work},
    booktitle = {Financial Cryptography and Data Security - 24th International Conference, {FC} 2020},
    series    = {Lecture Notes in Computer Science},
    volume    = {12059},
    pages     = {505--522},
    publisher = {Springer},
    year      = {2020},
    doi       = {10.1007/978-3-030-51280-4\_27}
}

@misc{plonky2,
    author       = {Polygon Zero Team},
    title        = {Plonky2: Fast Recursive Arguments with PLONK and FRI},
    year         = {2022},
    howpublished = {\url{https://github.com/0xPolygonZero/plonky2/blob/main/plonky2/plonky2.pdf}},
    note         = {Accessed: 2025-03-13}
}

@inproceedings{valiant2008incrementally,
    title        = {Incrementally verifiable computation or proofs of knowledge imply time/space efficiency},
    author       = {Valiant, Paul},
    booktitle    = {Theory of Cryptography: Fifth Theory of Cryptography Conference, TCC 2008, New York, USA, March 19-21, 2008. Proceedings 5},
    pages        = {1--18},
    year         = {2008},
    organization = {Springer}
}

@inproceedings{recursion-cycles,
    author    = {Ben-Sasson, Eli
                 and Chiesa, Alessandro
                 and Tromer, Eran
                 and Virza, Madars},
    editor    = {Garay, Juan A.
                 and Gennaro, Rosario},
    title     = {Scalable Zero Knowledge via Cycles of Elliptic Curves},
    booktitle = {Advances in Cryptology -- CRYPTO 2014},
    year      = {2014},
    publisher = {Springer Berlin Heidelberg},
    address   = {Berlin, Heidelberg},
    pages     = {276--294},
    isbn      = {978-3-662-44381-1}
}

@inproceedings{bitansky2013,
    author    = {Bitansky, Nir and Canetti, Ran and Chiesa, Alessandro and Tromer, Eran},
    title     = {Recursive composition and bootstrapping for SNARKS and proof-carrying data},
    year      = {2013},
    isbn      = {9781450320290},
    publisher = {Association for Computing Machinery},
    address   = {New York, NY, USA},
    url       = {https://doi.org/10.1145/2488608.2488623},
    doi       = {10.1145/2488608.2488623},
    booktitle = {Proceedings of the Forty-Fifth Annual ACM Symposium on Theory of Computing},
    pages     = {111–120},
    numpages  = {10},
    location  = {Palo Alto, California, USA},
    series    = {STOC '13}
}

@misc{cryptoeprint:2020/1522,
    author       = {Weikeng Chen and Alessandro Chiesa and Emma Dauterman and Nicholas P.  Ward},
    title        = {Reducing Participation Costs via Incremental Verification for Ledger Systems},
    howpublished = {Cryptology {ePrint} Archive, Paper 2020/1522},
    year         = {2020},
    url          = {https://eprint.iacr.org/2020/1522}
}

@inproceedings{flyclient,
    author    = {Bünz, Benedikt and Kiffer, Lucianna and Luu, Loi and Zamani, Mahdi},
    booktitle = {2020 IEEE Symposium on Security and Privacy (SP)},
    title     = {FlyClient: Super-Light Clients for Cryptocurrencies},
    year      = {2020},
    volume    = {},
    number    = {},
    pages     = {928-946},
    doi       = {10.1109/SP40000.2020.00049}
}

@misc{tokio,
    author       = {The Tokio Team},
    title        = {Tokio},
    howpublished = {\url{https://tokio.rs}},
    year         = 2024
}

@misc{halo,
      author = {Sean Bowe and Jack Grigg and Daira Hopwood},
      title = {Recursive Proof Composition without a Trusted Setup},
      howpublished = {Cryptology {ePrint} Archive, Paper 2019/1021},
      year = {2019},
      url = {https://eprint.iacr.org/2019/1021}
}

@misc{plonky3,
      author = {Polygon Zero Team},
      title = {Plonky3 Repository},
      howpublished = {\url{https://github.com/Plonky3/Plonky3}},
      year = {2025}
}

@misc{risc0,
      author = {RISC Zero Team},
      title = {RISC Zero Repository},
      howpublished = {\url{https://github.com/risc0/risc0}},
      year = {2025}
}

@inproceedings{al2018fraud,
  author    = {Mustafa Al{-}Bassam and Alberto Sonnino and Vitalik Buterin and Ismail Khoffi},
  title     = {Fraud and Data Availability Proofs: Detecting Invalid Blocks in Light Clients},
  booktitle = {Financial Cryptography and Data Security - 25th International Conference, {FC} 2021, Part {II}},
  series    = {Lecture Notes in Computer Science},
  volume    = {12675},
  pages     = {279--298},
  publisher = {Springer},
  year      = {2021},
  doi       = {10.1007/978-3-662-64331-0\_15}
}

@inproceedings{plumo,
  author    = {Psi Vesely and Kobi Gurkan and Michael Straka and Ariel Gabizon and Philipp Jovanovic and Georgios Konstantopoulos and Asa Oines and Marek Olszewski and Eran Tromer},
  title     = {Plumo: An Ultralight Blockchain Client},
  booktitle = {Financial Cryptography and Data Security - 26th International Conference, {FC} 2022},
  series    = {Lecture Notes in Computer Science},
  volume    = {13411},
  pages     = {597--614},
  publisher = {Springer},
  year      = {2022},
  doi       = {10.1007/978-3-031-18283-9\_30}
}

@inproceedings{nova,
  author    = {Abhiram Kothapalli and Srinath T. V. Setty and Ioanna Tzialla},
  title     = {Nova: Recursive Zero-Knowledge Arguments from Folding Schemes},
  booktitle = {Advances in Cryptology - {CRYPTO} 2022, Part {IV}},
  series    = {Lecture Notes in Computer Science},
  volume    = {13510},
  pages     = {359--388},
  publisher = {Springer},
  year      = {2022},
  doi       = {10.1007/978-3-031-15985-5\_13}
}

@inproceedings{hypernova,
  author    = {Abhiram Kothapalli and Srinath T. V. Setty},
  title     = {HyperNova: Recursive Arguments for Customizable Constraint Systems},
  booktitle = {Advances in Cryptology - {CRYPTO} 2024, Part {X}},
  series    = {Lecture Notes in Computer Science},
  volume    = {14929},
  pages     = {345--379},
  publisher = {Springer},
  year      = {2024},
  doi       = {10.1007/978-3-031-68403-6\_11}
}

@inproceedings{protostar,
  author    = {Benedikt B{\"{u}}nz and Binyi Chen},
  title     = {Protostar: Generic Efficient Accumulation/Folding for Special-Sound Protocols},
  booktitle = {Advances in Cryptology - {ASIACRYPT} 2023, Part {II}},
  series    = {Lecture Notes in Computer Science},
  volume    = {14439},
  pages     = {77--110},
  publisher = {Springer},
  year      = {2023},
  doi       = {10.1007/978-981-99-8724-5\_3}
}

@inproceedings{lazyclient,
  author    = {Ertem Nusret Tas and David Tse and Lei Yang and Dionysis Zindros},
  title     = {Light Clients for Lazy Blockchains},
  booktitle = {Financial Cryptography and Data Security - 28th International Conference, {FC} 2024, Part {II}},
  series    = {Lecture Notes in Computer Science},
  volume    = {14745},
  pages     = {3--21},
  publisher = {Springer},
  year      = {2025},
  doi       = {10.1007/978-3-031-78679-2\_1}
}

@misc{sunfish,
  author       = {Giulia Scaffino and Philipp Slowak and Karl W{\"u}st and Deepak Maram and Alberto Sonnino and Lefteris Kokoris-Kogias},
  title        = {Sunfish: Reading Ledgers with Sparse Nodes},
  howpublished = {Cryptology {ePrint} Archive, Paper 2024/1680},
  year         = {2024},
  url          = {https://eprint.iacr.org/2024/1680}
}

@inproceedings{PSG+24,
  author    = {Charalampos Papamanthou and Shravan Srinivasan and Nicolas Gailly and Ismael Hishon{-}Rezaizadeh and Andrus Salumets and Stjepan Golemac},
  title     = {Reckle Trees: Updatable Merkle Batch Proofs with Applications},
  booktitle = {Proceedings of the 2024 on {ACM} {SIGSAC} Conference on Computer and Communications Security, {CCS} 2024},
  pages     = {1538--1551},
  publisher = {{ACM}},
  year      = {2024},
  doi       = {10.1145/3658644.3670354}
}

@inproceedings{zkbridge,
  author    = {Tiancheng Xie and Jiaheng Zhang and Zerui Cheng and Fan Zhang and Yupeng Zhang and Yongzheng Jia and Dan Boneh and Dawn Song},
  title     = {zkBridge: Trustless Cross-chain Bridges Made Practical},
  booktitle = {Proceedings of the 2022 {ACM} {SIGSAC} Conference on Computer and Communications Security, {CCS} 2022},
  pages     = {3003--3017},
  publisher = {{ACM}},
  year      = {2022},
  doi       = {10.1145/3548606.3560652}
}

@misc{altair-light-client,
  author       = {{Ethereum Foundation}},
  title        = {Altair Light Client -- Sync Protocol},
  howpublished = {Ethereum consensus specifications},
  url          = {https://ethereum.github.io/consensus-specs/specs/altair/light-client/sync-protocol/},
  note         = {Accessed September 2026}
}

@misc{helios,
  author       = {{a16z crypto}},
  title        = {Helios: A fast, secure, and portable multichain light client for Ethereum},
  howpublished = {\url{https://github.com/a16z/helios}},
  year         = {2022},
  note         = {Accessed September 2026}
}

@misc{jellyfish,
  author       = {Zhenhuan Gao and Yuxuan Hu and Qinfan Wu},
  title        = {Jellyfish Merkle Tree},
  howpublished = {Diem technical paper},
  year         = {2021},
  url          = {https://developers.diem.com/papers/jellyfish-merkle-tree/2021-01-14.pdf}
}

@inproceedings{lvmt,
  author    = {Chenxing Li and Sidi Mohamed Beillahi and Guang Yang and Ming Wu and Wei Xu and Fan Long},
  title     = {{LVMT:} An Efficient Authenticated Storage for Blockchain},
  booktitle = {17th {USENIX} Symposium on Operating Systems Design and Implementation, {OSDI} 2023},
  pages     = {135--153},
  publisher = {{USENIX} Association},
  year      = {2023},
  url       = {https://www.usenix.org/conference/osdi23/presentation/li-chenxing}
}

@inproceedings{asvc,
  author    = {Alin Tomescu and Ittai Abraham and Vitalik Buterin and Justin Drake and Dankrad Feist and Dmitry Khovratovich},
  title     = {Aggregatable Subvector Commitments for Stateless Cryptocurrencies},
  booktitle = {Security and Cryptography for Networks - 12th International Conference, {SCN} 2020},
  series    = {Lecture Notes in Computer Science},
  volume    = {12238},
  pages     = {45--64},
  publisher = {Springer},
  year      = {2020},
  doi       = {10.1007/978-3-030-57990-6\_3}
}

@inproceedings{hyperproofs,
  author    = {Shravan Srinivasan and Alexander Chepurnoy and Charalampos Papamanthou and Alin Tomescu and Yupeng Zhang},
  title     = {Hyperproofs: Aggregating and Maintaining Proofs in Vector Commitments},
  booktitle = {31st {USENIX} Security Symposium, {USENIX} Security 2022},
  pages     = {3001--3018},
  publisher = {{USENIX} Association},
  year      = {2022},
  url       = {https://www.usenix.org/conference/usenixsecurity22/presentation/srinivasan}
}

@misc{sui-checkpoint-verification,
  author       = {{Mysten Labs}},
  title        = {Checkpoint Verification -- Sui Documentation},
  howpublished = {\url{https://docs.sui.io/concepts/cryptography/system/checkpoint-verification}},
  note         = {Accessed September 2026}
}

@misc{solana-simd-0215,
  author       = {{Solana Foundation}},
  title        = {{SIMD-0215}: Homomorphic Hashing of Account State (Accounts Lattice Hash)},
  howpublished = {\url{https://github.com/solana-foundation/solana-improvement-documents/blob/main/proposals/0215-accounts-lattice-hash.md}},
  year         = {2024},
  note         = {Created 2024-12-20; status: activated. Accessed September 2026}
}

@misc{lagrange-coprocessor,
  author       = {{Lagrange Labs}},
  title        = {Lagrange {ZK} Coprocessor and Verifiable Database (Euclid)},
  howpublished = {\url{https://www.lagrange.dev/blog/announcing-testnet-euclid-zk-coprocessor}},
  year         = {2024},
  note         = {Accessed September 2026}
}

@inproceedings{zk-rollups-bench,
  author    = {Stefanos Chaliasos and Itamar Reif and Adri{\`{a}} Torralba{-}Agell and Jens Ernstberger and Assimakis Kattis and Benjamin Livshits},
  title     = {Analyzing and Benchmarking ZK-Rollups},
  booktitle = {6th Conference on Advances in Financial Technologies, {AFT} 2024},
  series    = {LIPIcs},
  volume    = {316},
  pages     = {6:1--6:24},
  publisher = {Schloss Dagstuhl - Leibniz-Zentrum f{\"{u}}r Informatik},
  year      = {2024},
  doi       = {10.4230/LIPIcs.AFT.2024.6}
}

@inproceedings{mangrove,
  author    = {Wilson D. Nguyen and Trisha Datta and Binyi Chen and Nirvan Tyagi and Dan Boneh},
  title     = {Mangrove: {A} Scalable Framework for Folding-Based SNARKs},
  booktitle = {Advances in Cryptology - {CRYPTO} 2024, Part {X}},
  series    = {Lecture Notes in Computer Science},
  volume    = {14929},
  pages     = {308--344},
  publisher = {Springer},
  year      = {2024},
  doi       = {10.1007/978-3-031-68403-6\_10}
}

@misc{zktree,
  author       = {Sai Deng and Bo Du},
  title        = {zkTree: A Zero-Knowledge Recursion Tree with {ZKP} Membership Proofs},
  howpublished = {Cryptology {ePrint} Archive, Paper 2023/208},
  year         = {2023},
  url          = {https://eprint.iacr.org/2023/208}
}

@inproceedings{dizk,
  author    = {Howard Wu and Wenting Zheng and Alessandro Chiesa and Raluca Ada Popa and Ion Stoica},
  title     = {{DIZK:} {A} Distributed Zero Knowledge Proof System},
  booktitle = {27th {USENIX} Security Symposium, {USENIX} Security 2018},
  pages     = {675--692},
  publisher = {{USENIX} Association},
  year      = {2018},
  url       = {https://www.usenix.org/conference/usenixsecurity18/presentation/wu}
}

@inproceedings{pianist,
  author    = {Tianyi Liu and Tiancheng Xie and Jiaheng Zhang and Dawn Song and Yupeng Zhang},
  title     = {Pianist: Scalable zkRollups via Fully Distributed Zero-Knowledge Proofs},
  booktitle = {{IEEE} Symposium on Security and Privacy, {SP} 2024},
  pages     = {1777--1793},
  publisher = {{IEEE}},
  year      = {2024},
  doi       = {10.1109/SP54263.2024.00035}
}

@inproceedings{hekaton,
  author    = {Michael Rosenberg and Tushar Mopuri and Hossein Hafezi and Ian Miers and Pratyush Mishra},
  title     = {Hekaton: Horizontally-Scalable zkSNARKs Via Proof Aggregation},
  booktitle = {Proceedings of the 2024 on {ACM} {SIGSAC} Conference on Computer and Communications Security, {CCS} 2024},
  pages     = {929--940},
  publisher = {{ACM}},
  year      = {2024},
  doi       = {10.1145/3658644.3690282}
}

@inproceedings{snarkpack,
  author    = {Nicolas Gailly and Mary Maller and Anca Nitulescu},
  title     = {SnarkPack: Practical {SNARK} Aggregation},
  booktitle = {Financial Cryptography and Data Security - 26th International Conference, {FC} 2022},
  series    = {Lecture Notes in Computer Science},
  volume    = {13411},
  pages     = {203--229},
  publisher = {Springer},
  year      = {2022},
  doi       = {10.1007/978-3-031-18283-9\_10}
}

@misc{sp1,
  author       = {Succinct Labs},
  title        = {{SP1}: A performant, open-source {zkVM}},
  howpublished = {\url{https://github.com/succinctlabs/sp1}},
  year         = {2025},
  note         = {Accessed September 2026}
}

@misc{verkle-buterin,
  author       = {Vitalik Buterin},
  title        = {Verkle trees},
  howpublished = {\url{https://vitalik.eth.limo/general/2021/06/18/verkle.html}},
  year         = {2021},
  note         = {Blog post, June 18, 2021. Accessed September 2026}
}

@misc{axiom,
  author       = {{Axiom}},
  title        = {Axiom {V2} Developer Docs: App Architecture},
  howpublished = {\url{https://docs.axiom.xyz/docs/axiom-developer-flow/app-architecture}},
  year         = {2024},
  note         = {Archived at \url{https://web.archive.org/web/20240627181800/https://docs.axiom.xyz/docs/axiom-developer-flow/app-architecture}}
}

@misc{brevis,
  author       = {Mo Dong and Qingkai Liang and Xiaozhou Li and Junda Liu},
  title        = {Brevis: An Omnichain {ZK} Data Attestation Platform},
  howpublished = {Whitepaper v1.0, Celer Network},
  year         = {2023},
  url          = {https://get.celer.app/brevis/BrevisWhitePaper_03211833.pdf}
}

@misc{ethereum,
  author       = {Vitalik Buterin},
  title        = {Ethereum: A Next-Generation Smart Contract and Decentralized Application Platform},
  howpublished = {Whitepaper, \url{https://ethereum.org/en/whitepaper/}},
  year         = {2014}
}

@misc{bonneau2020mina,
  author       = {Joseph Bonneau and Izaak Meckler and Vanishree Rao and Evan Shapiro},
  title        = {Mina: Decentralized Cryptocurrency at Scale},
  howpublished = {Whitepaper, O(1) Labs},
  year         = {2020},
  url          = {https://minaprotocol.com/wp-content/uploads/technicalWhitepaper.pdf}
}

@misc{solidity-events,
  title        = {Solidity Documentation 0.8.29: Contracts -- Events},
  author       = {{Solidity Team}},
  howpublished = {\url{https://docs.soliditylang.org/en/v0.8.29/contracts.html#events}},
  year         = {2025}
}

@inproceedings{sui-lutris,
  author    = {Sam Blackshear and Andrey Chursin and George Danezis and Anastasios Kichidis and Lefteris Kokoris-Kogias and Xun Li and Mark Logan and Ashok Menon and Todd Nowacki and Alberto Sonnino and Brandon Williams and Lu Zhang},
  title     = {Sui Lutris: {A} Blockchain Combining Broadcast and Consensus},
  booktitle = {Proceedings of the 2024 on {ACM} {SIGSAC} Conference on Computer and Communications Security, {CCS} 2024},
  pages     = {2606--2620},
  publisher = {{ACM}},
  year      = {2024},
  doi       = {10.1145/3658644.3670286}
}

@inproceedings{blink,
  author    = {Lukas Aumayr and Zeta Avarikioti and Matteo Maffei and Giulia Scaffino and Dionysis Zindros},
  title     = {Blink: An Optimal Proof of Proof-of-Work},
  booktitle = {Financial Cryptography and Data Security - 29th International Conference, {FC} 2025},
  series    = {Lecture Notes in Computer Science},
  volume    = {15752},
  pages     = {173--190},
  publisher = {Springer},
  year      = {2025},
  doi       = {10.1007/978-3-032-07035-7\_11}
}

@misc{sui-snapshots,
  author       = {{Mysten Labs}},
  title        = {Database Snapshots -- Sui Documentation},
  howpublished = {\url{https://docs.sui.io/operators/snapshots}},
  note         = {Accessed September 2026}
}

\appendix
\renewcommand{\theHsection}{\Alph{section}}
\crefalias{section}{appendix}
\crefalias{subsection}{appendix}
\section{Proof of \Cref{thm:minnow-security}}
\label{sec:minnow-a-proof}

This appendix proves \Cref{thm:minnow-security}: if a client accepts
${\sf verifyState}(\streamid, \cp, \streamstate, \pi)$ against a blockchain
$B$, then $(\streamid, \streamstate)$ is correct at checkpoint~$\cp$ in the
sense of \Cref{def:stream-correctness}, under the unforgeability of the
signature scheme, the collision resistance of the hash function, and the
knowledge soundness of the proof system. The proof is a sequence of hybrid
games, each removing one way for the adversary to win.

\begin{proof}
  Let an adversary $\mathcal{A}$ output $(\streamid, c, (v,e), \pi)$ and
  auxiliary data such that the client accepts, but $(\streamid, v)$
  is \emph{incorrect} at~$c$. We upper bound $\Pr[\text{$\mathcal{A}$
    wins}]$ by the advantages of breaking our assumptions.

  \paragraph{Game 0 (real).}
  This is the real experiment underlying \Cref{thm:minnow-security}:
  the verifier runs ${\sf verifyState}(\streamid, c, (v,e), \pi)$
  exactly as specified. By definition, $\Pr[\text{win in $G_0$}] =
    \Pr[\text{$\mathcal{A}$ wins}]$.

  \paragraph{Game 1 (header authentication).}
  In this game, the verifier rejects if the current committee did not
  sign the header~$\header{c}$ of $B$, but the signature verification
  succeeds. Since the client holds the authentic committee
  key~$\validatorpk$ (committee handoffs are authenticated from genesis,
  \Cref{sec:model}), any such header is a forgery, so
  \[
    |\Pr[\text{win in $G_0$}] - \Pr[\text{win in $G_1$}]| \leq
    \mathsf{Adv}^{\sf sig}_{\mathcal{A}}.
  \]

  \paragraph{Game 2 (knowledge soundness).}
  The verifier checks ${\sf ZK.verify}(\updateproof{c}, \{\globaldigest{c},
  \localdigest{c}, \numleaves{c}\})$ but aborts if the extractor fails to
  recover a valid witness. By knowledge soundness of the employed ZK system
  for all involved statements (including recursive verification), there
  exists a polynomial-time extractor $\mathcal{E}$ that, on any accepting
  transcript, outputs a concrete execution trace and witnesses: a sequence
  of public inputs $(\globaldigest{t}, \localdigest{t}, \numleaves{t})$
  for all $t = 0,1,\ldots,c$, with $(\globaldigest{0}, \localdigest{0},
  \numleaves{0}) = (\globaldigest{\sf init}, \localdigest{\sf init},
  \numleaves{\sf init})$ enforced by the circuit base case
  (\Cref{app:circuit-a}); for each checkpoint $t \in [1..c]$, the batch
  contents $\chainupdates{t}$ such that the circuit relation $\localdigest{t} =
  \minnowhash(\localdigest{t-1}, \chainupdates{t})$ holds; and for each
  $(\streamid_j,\streamstate'_j,\streamnewbit_j) \in \chainupdates{t}$, Merkle
  witnesses showing that ${\sf MT.vfAndUpd}$ transforms
  $(\globaldigest{t-1},\streampoint_{\streamid_j})$ into $\globaldigest{t}$
  at the correct index, with the per-update consistency checks enforced by
  the circuit. Consequently, conditioned on knowledge soundness,
  $(\globaldigest{c}, \localdigest{c}, \numleaves{c})$ are uniquely
  determined by the extracted updates $\{\chainupdates{t}\}_{t=1}^c$ and the
  genesis inputs. Any acceptance with an incorrect $(\streamid,v)$ can only
  persist if we later find a collision in $\minnowhash$ or in the Merkle
  hash (next game), or a bad Merkle inclusion (final game). Therefore,
  \[
    |\Pr[\text{win in $G_1$}] - \Pr[\text{win in $G_2$}]| \leq
    \mathsf{Adv}^{\sf ksnd}_{\mathcal{A}}.
  \]

  \paragraph{Game 3 (collision resistance).}
  Relative to $G_2$, we conceptually replace the hash computations by
  ideal bindings: we require, and check in the analysis, that
  $\localdigest{t} = \minnowhash(\localdigest{t-1}, \chainupdates{t})$ and that
  ${\sf MT.vfAndUpd}$ yields the stated $\globaldigest{t}$ for all $t$. In
  the extracted trace from $G_2$ these relations already hold. An adversary
  that changes batch contents or leaf updates while keeping the same digests
  produces either a checkpoint whose committed batch contents differ under
  the same $\localdigest{\cdot}$, a collision in $\minnowhash$, or a Merkle
  update or inclusion altered under the same root, a collision in the Merkle
  hash. Hence $|\Pr[\text{win in $G_2$}] - \Pr[\text{win in $G_3$}]| \leq
  \mathsf{Adv}^{\sf crhf}_{\mathcal{A}}$.

  \paragraph{No win in the final game.}
  In $G_3$, (i) the signed header $\header{c}$ of $B$ binds
  $\localdigest{c}$, (ii) the extracted trace and collision resistance fix
  the unique sequence of updates from genesis yielding $\globaldigest{c}$,
  and (iii) the queried $(\streamid,v,e)$ is a leaf under
  $\globaldigest{c}$. Therefore $v$ equals the accumulator over all updates
  to $\streamid$ up to $c$, so the adversary cannot win. By a telescoping
  sum over the games,
  \[
    \Pr[\text{$\mathcal{A}$ wins}] \leq \mathsf{Adv}^{\sf
      sig}_{\mathcal{A}} + \mathsf{Adv}^{\sf ksnd}_{\mathcal{A}} +
    \mathsf{Adv}^{\sf crhf}_{\mathcal{A}} + \mathsf{negl}(\lambda),
  \]
  which is negligible under the stated assumptions.
\end{proof}

\section{Formal statements}
\label{app:formal}

This appendix collects the formal material that \Cref{sec:model} and
\Cref{sec:minnow-mt-a,sec:minnow-mt-b} summarize: the Merkle-tree
operations, the full statement of the \minnowMTA{} circuit, the procedure
for serving and verifying proofs, and the statement of the cyclic circuit of
\minnowMTB{}.

\subsection{Merkle-tree API}
\label{app:merkle}

The \sysname{} protocols use the following Merkle-tree operations.

\begin{itemize}
  \item $0/1 \gets {\sf MT.vf}(\streamid, \streamstate,
    \globaldigest, \merklepath)$:
    Verify that the value~$(\streamid, \streamstate)$ is a leaf
    in a Merkle tree with root~$\globaldigest$ using the provided
    path~$\merklepath$.
  \item $D' \gets {\sf MT.vfAndUpd}(\streamid, \streamstate, \streamstate',
    \globaldigest, \merklepath)$: Verify inclusion of $(\streamid,
    \streamstate)$ using the path like above.
    Then, update the stream state to the new value~$\streamstate'$
    along the same path and return the new root.
  \item $D' \gets {\sf MT.insert}(\streamid, \streamstate, \globaldigest, idx)$:
    Insert the value~$(\streamid, \streamstate)$ at the given
    index~$idx$ in the Merkle tree
    with root~$\globaldigest$ and return the new root.
\end{itemize}

\subsection{The \minnowMTA{} circuit}
\label{app:circuit-a}
At setup time, we fix the following parameters and hard-code them
into the circuit.\footnote{Unlike many past cryptographic constructions,
  cheap recursion allows an easy way to change parameters: a new circuit with
  different parameters can simply verify the old proof in its base case, and
new proofs thereafter.}
\begin{enumerate}
  \item \textbf{Tree depth}~$\treedepth$: determines the maximum
    number of stream IDs supported, $\maxnumstreams = 2^{\treedepth}$ (e.g.,
    $\treedepth = 30$ for $M = 2^{30}$ stream IDs).
  \item \textbf{Batch size}~$\batchsize$: the batch size used by validators.
  \item \textbf{Batches per bundle}~$\numcircuitbatches$: the number
    of batches processed
    by a single circuit invocation (leading to a bundle of~$\numcircuitupdates =
    \numcircuitbatches \cdot \batchsize$ updates).
\end{enumerate}

\paragraph{Circuit inputs and Merkle-tree update.}
Updating the Merkle tree and deriving the circuit inputs are
intertwined, so we describe them together.

The circuit processes a bundle of stream updates at a time denoted as
$\circuitbundle{} = [\streampoint'_{\streamid_1}, \ldots,
\streampoint'_{\streamid_{\numcircuitupdates}}]$.
In theory, this could represent the updates of a single checkpoint,
or even a batch of checkpoints.
Suppose the Merkle tree currently represents the state after
checkpoint~$\cp-1$; we set $\globaldigest{} \gets \globaldigest{\cp-1}$.
For each update $\streampoint'_{\streamid_j} = (\streamid_j,
\streamstate'_j, \streamnewbit_j)$ in the bundle~$\circuitbundle{}$:
\begin{enumerate}
  \item If $\streamnewbit_j = 0$, retrieve leaf $(\streamid_j,
    \streamstate_j)$ and its Merkle path
    $\merklepath_{\streamid_j}$, and update:
    \[
      \globaldigest{} \gets
      {\sf MT.vfAndUpd}(
        \streamid_j,
        \streamstate_j,
        \streamstate'_j,
        \globaldigest{},
      \merklepath_{\streamid_j}).
    \]
  \item If $\streamnewbit_j = 1$, insert
    $(\streamid_j, \streamstate'_j)$ at the index
    $\numleaves{}$, and compute its corresponding path
    $\merklepath_{\streamid_j}$:
    \[
      \globaldigest{} \gets
      {\sf MT.insert}(
        \streamid_j,
        \streamstate'_j,
        \globaldigest{},
      \numleaves{}).
    \]
\end{enumerate}
The tuples
$\{ \streampoint_{\streamid_j}, \streampoint'_{\streamid_j},
\merklepath_{\streamid_j} \}$
(for new streams, $\streampoint_{\streamid_j} = \bot$)
constitute the private inputs to the ZK circuit.
If~$\circuitbundle{}$ corresponds to the updates of a single checkpoint~$\cp$,
then the above process computes the new Merkle root after checkpoint~$\cp$,
i.e., $\globaldigest{} = \globaldigest{\cp}$.

\paragraph{Circuit statement ($\minnowMTAstatement$).}

We now specify the \minnowMTA{} ZK circuit's statement~$\minnowMTAstatement$.
In the base case, it accepts $\updateproof{\cp'}=\bot$ and sets all the
public inputs to their initial values: $\cp_{\sf init} = 0$,
$\globaldigest{\sf init} = 0$, $\localdigest{\sf init} = 0$,
$\numleaves{\sf init} = 0$.

\begin{itemize}
  \item \textbf{Private inputs:}
    \begin{enumerate}
      \item Previous proof $\updateproof{\cp'}$ with its public inputs
        $\pubinputs{\cp'} = \{\globaldigest{\cp'},
        \localdigest{\cp'}, \numleaves{\cp'}, \cp'\}$
        (set to $\bot$ in the base case);
      \item Updates:
        $\{(\streampoint_{\streamid_j}, \streampoint'_{\streamid_j},
        \merklepath_{\streamid_j})\}_{j \in [\numcircuitupdates]}$.
    \end{enumerate}
  \item \textbf{Public inputs:}
    $\pubinputs{\cp} = \{\globaldigest{\cp}, \localdigest{\cp},
    \numleaves{\cp}, \cp\}$.

  \item \textbf{Statement logic:}
    \begin{enumerate}
      \item \emph{Base case:}
        If $\cp = \cp_{\sf init}$, set
        $\globaldigest{\cp} \gets \globaldigest{\sf init}$,
        $\localdigest{\cp} \gets \localdigest{\sf init}$,
        $\numleaves{\cp} \gets \numleaves{\sf init}$,
        and accept. Else assert $\cp = \cp' + 1$.
      \item \emph{Verify prior proof:}
        ${\sf ZK.verify}(\updateproof{\cp'}, \pubinputs{\cp'})$.
      \item \emph{Update Merkle tree:}
        Initialize
        $\globaldigest{} \gets \globaldigest{\cp'}$, $\numleaves{}
        \gets \numleaves{\cp'}$.
        For each update $j \in [\numcircuitupdates]$:
        \begin{itemize}
          \item Update root:
            $\globaldigest{} \gets {\sf MT.vfAndUpd}(\streamid_j,
              \streamstate_j, \streamstate'_j, \globaldigest{},
            \merklepath_{\streamid_j})$;
          \item If $\streamnewbit_j = 1$, assert
            $\merklepath_{\streamid_j}.\text{index} = \numleaves{}$
            and increment $\numleaves{}$ by one.
          \item If $\streamnewbit_j = 0$, assert
            $\streampoint_{\streamid_j}.\streamid =
            \streampoint'_{\streamid_j}.\streamid$.
        \end{itemize}
      \item \emph{Checks:}
        Assert
        $\globaldigest{} = \globaldigest{\cp}$,
        $\numleaves{} = \numleaves{\cp}$, and
        $\localdigest{\cp} =
        \minnowhash(\localdigest{\cp'},
        \{\streampoint'_{\streamid_j}\}_{j \in [\numcircuitupdates]})$.
    \end{enumerate}
\end{itemize}

The circuit outputs the updated proof $\updateproof{\cp}$ attesting to
the correctness of the new global state commitment~$\globaldigest{\cp}$.

\subsection{Serving and verifying proofs}
\label{app:client}
Suppose a client queries the state of a stream identifier~$\streamid$ at
checkpoint~$\cp$.
The ZK service computes
$(\streamstate, \validityproof)
\gets {\sf readState}(\streamid, \cp)$,
where $\validityproof = \{\updateproof{\cp}, \header{\cp},
\merklepath_{\streamid}\}$.
The proof consists of three components:

\begin{itemize}
  \item The recursive ZK proof~$\updateproof{\cp}$ with its public
    inputs $\{\globaldigest{\cp}, \localdigest{\cp}, \numleaves{\cp}, \cp\}$.
  \item The checkpoint header~$\header{\cp}$, which includes the
    stream update commitment~$\localdigest{\cp}$ and a validator
    signature over the header (see \Cref{sec:model}).
  \item The Merkle path~$\merklepath_{\streamid}$ for the queried stream.
\end{itemize}

The client verifies the proof using the committee keys~$\validatorpk$
(\Cref{sec:model}) as follows:
\begin{enumerate}
  \item Check the Merkle path:
    ${\sf MT.vf}(\streamid, \streamstate, \globaldigest{\cp},
    \merklepath_{\streamid})$.
  \item Verify the validator-signed checkpoint header:
    ${\sf Sig.verify}(\header{\cp}, \validatorpk)$,
    and ensure that
    $\header{\cp}.\localdigest{} = \localdigest{\cp}$
    and
    $\header{\cp}.\cp = \cp$.
  \item Verify the ZKP:
    ${\sf ZK.verify}(\updateproof{\cp}, \{\globaldigest{\cp},
    \localdigest{\cp}, \numleaves{\cp}, \cp\})$.
\end{enumerate}

This procedure corresponds to the API call
${\sf verifyState}(\streamid, \cp, \streamstate, \validityproof)$.

\subsection{The cyclic circuit of \minnowMTB{}}
\label{app:cyclic}
The statement~$\sequentialphasestatement$ of the cyclic circuit is as follows:
\begin{itemize}
\item \textbf{Private inputs:} (i) last proof and inputs:
$\sequentialphaseproof{\cp'}$, $\sequentialphasepubinputs{\cp'} =
\{\globaldigest{\cp'}, \localdigest{\cp'}, \numleaves{\cp'}, \cp'\}$,
(ii) aggregated proof for checkpoint $\cp$:
$\intermediatephaseproof{\cp'}{\cp}$.
\item \textbf{Public inputs:}
$\sequentialphasepubinputs{\cp}=\{\globaldigest{\cp}, \newfield{\cp},
\numleaves{\cp}, \cp\}$.
\item \textbf{Statement:}
\begin{enumerate}
\item Verify the previous sequential proof:
  ${\sf ZK.verify}(\sequentialphaseproof{\cp'},
  \sequentialphasepubinputs{\cp'})$ and check $\cp = \cp' + 1$.
\item Construct public inputs for
  aggregation circuit: $\intermediatephasepubinputs{\cp'}{\cp} =
  \{\globaldigest{\cp'}, \globaldigest{\cp}, \localdigest{\cp'},
  \localdigest{\cp}, \numleaves{\cp'}, \numleaves{\cp}\}$ and verify
  ${\sf ZK.verify}(\intermediatephaseproof{\cp'}{\cp},
  \intermediatephasepubinputs{\cp'}{\cp})$.
\end{enumerate}
\end{itemize}

\section{\hashprotocol{}}
\label{sec:minnow-h}

We now describe a lightweight variant of \sysname{} in which the ZK
service tracks only a \emph{fixed, small set} of stream IDs
$\streamidset = \{\streamid_1,\ldots,\streamid_m\}$, rather than all
streams on the chain. This is useful when an application cares about
completeness only for a few selected streams and wishes to minimize
proving costs.

The ZK service maintains an internal state
$\globalstatemap{\cp} = [(\streamid_i,\streamstate_i)]_{i=1}^m$
after checkpoint~$\cp$, and commits to it using a simple hash
\[
  \globaldigest{\cp} = H(\globalstatemap{\cp}).
\]
Validators, as in the main protocols, compute a per-checkpoint
hash-chain commitment~$\localdigest{\cp}$ over the (sorted) stream
updates they derive from checkpoint contents, and include
$\localdigest{\cp}$ in the checkpoint header.

At each checkpoint, the recursive circuit for \hashprotocol{} takes as
input:
(i) the previous proof and its public inputs
$(\globaldigest{\cp-1}, \localdigest{\cp-1})$,
(ii) the validator-supplied stream updates for this checkpoint, and
(iii) the previous state $\globalstatemap{\cp-1}$.
It checks that:
\begin{itemize}
  \item the prior proof is valid and consistent with
    $\globaldigest{\cp-1}$ and $\localdigest{\cp-1}$,
  \item applying the checkpoint’s updates to $\globalstatemap{\cp-1}$
    yields a new state $\globalstatemap{\cp}$ consistent with the
    tracked stream set~$\streamidset$, and
  \item the public outputs satisfy
    $\globaldigest{\cp} = H(\globalstatemap{\cp})$ and
    $\localdigest{\cp} = \minnowhash(\localdigest{\cp-1},
    \chainupdates{\cp})$.
\end{itemize}
A client that wants the state of some $\streamid \in \streamidset$ at
checkpoint~$\cp$ receives:
(i) the claimed value~$\streamstate$,
(ii) the full list $\globalstatemap{\cp}$ (of size~$m$),
(iii) the recursive proof for~$\cp$, and
(iv) the signed checkpoint header.
To verify, the client (a) checks that $(\streamid,\streamstate)$
appears in $\globalstatemap{\cp}$, (b) recomputes
$H(\globalstatemap{\cp})$ and matches it with the proof’s
$\globaldigest{\cp}$, (c) verifies the recursive proof, and (d)
verifies the validator signature on the header and the embedded
$\localdigest{\cp}$.
For small~$m$, this yields short proofs and constant-time verification
for each query.

\section{Evaluation details}
\label{app:eval-details}

This appendix supplements \Cref{sec:eval} with the measurements and
derivations behind the numbers quoted there: the circuit sizes that fix the
batch and sub-bundle parameters (\Cref{app:circuit-sizes}), the number of
machines that \minnowMTB{} needs and the full special-configuration table
(\Cref{app:machines}), and the per-prover latency breakdown of the
end-to-end runs (\Cref{sec:breakdown}).

\subsection{Circuit sizes}
\label{app:circuit-sizes}

Two circuit-size measurements determine the parameters used in
\Cref{sec:eval}. \Cref{tab:batch-size-experiments} shows that the batch
size~$\batchsize$ barely changes the size of the base circuit, so it can be
chosen for the validators' convenience (\Cref{sec:validator-overhead}).
\Cref{tab:num-updates-par} gives the largest sub-bundle the base circuit of
\minnowMTB{} handles at each Plonky2 degree, which fixes the sub-bundle
size~$\subbundlesize$ of the special configuration.

\begin{table}[t]
  \begin{minipage}[t]{0.47\textwidth}
    \centering
    \small
    \begin{tabular}{lr}
      \toprule
      Batch size ($\batchsize$) & Gates \\
      \midrule
      10   & 50,000 \\
      100  & 49,890 \\
      1,000 & 49,876 \\
      \bottomrule
    \end{tabular}
    \caption{Size of the base circuit for varying batch sizes; the circuit
      processes $\parallelphasebatchsize = 1{,}000$ updates (batch size times
      number of batches).}
    \label{tab:batch-size-experiments}
  \end{minipage}\hfill
  \begin{minipage}[t]{0.47\textwidth}
    \centering
    \small
    \begin{tabular}{lr}
      \toprule
      Degree & Sub-bundle size \\
      \midrule
      14 & 200  \\
      15 & 450  \\
      16 & 900  \\
      17 & 1,900 \\
      \bottomrule
    \end{tabular}
    \caption{Maximum sub-bundle size ($\numcircuitbatches \cdot \batchsize$
      updates) that the base circuit~$\parallelphasestatement$ of \minnowMTB{}
      handles at each Plonky2 degree, tested in increments of fifty with
      $\maxnumstreams=2^{30}$ and $\batchsize = 50$.}
    \label{tab:num-updates-par}
  \end{minipage}
\end{table}

\subsection{Number of machines for \minnowMTB{}}
\label{app:machines}

To keep pace with the chain, \minnowMTB{} needs
\[
1
+ \numsubbundles{} \cdot \left(\frac{\parallelphasetime}{\bundletime}\right)
+ \numsubbundles{} \cdot \left(\frac{\intermediatephasetime}{\bundletime}\right)
\cdot \left(\frac{1 -
\tfrac{1}{\aggregationfactor^{\aggregationlevels}}}{\aggregationfactor
- 1}\right).
\]
machines: $\numsubbundles{} \cdot \parallelphasetime / \bundletime$ for the
parallel phase, $\numsubbundles{} \cdot (\tfrac{1}{\aggregationfactor} +
\tfrac{1}{\aggregationfactor^2} + \cdots +
\tfrac{1}{\aggregationfactor^{\aggregationlevels}}) \cdot
\intermediatephasetime / \bundletime$ for the aggregation phase, and one for
the sequential phase.

In the special configuration of \Cref{sec:minnow-b-microbenchmarks}, where $\numsubbundles{} = 3^\aggregationlevels$, this gives
\[
  1 + \beta + \beta (\tfrac{1 -
    \tfrac{1}{\aggregationfactor^{\aggregationlevels}}}{\aggregationfactor-1})
  = 1 + 3^\aggregationlevels + (3^\aggregationlevels - 1) / 2
    = (3^{\aggregationlevels + 1} + 1) / 2.
\]

This special configuration performed best among the parameter
choices we experimented with, but we do not claim it is globally optimal.
Choosing parameters for \minnowMTB{} is a multi-dimensional
optimization problem: for a given target throughput (which fixes the
bundle size~$\bundlesize$), one would like to pick the batch size,
sub-bundle size, and aggregation factor that minimize end-to-end latency.
A natural dual formulation is to fix a cap on the number of machines and
ask for the configuration that maximizes sustainable throughput.
We leave a systematic exploration of this design space to future work.

\begin{table}
  \centering
  \small
  \begin{tabular}{rrrrr}
    \toprule
    Updates ($\bundlesize$) & Levels ($\aggregationlevels$) &
    Latency ($\latency$, s) & Throughput (updates/s)        & Machines               \\
    \midrule
    1--200                  & 0                             & 1.26     & 317   & 2   \\
    201--600                & 1                             & 1.89     & 952   & 5   \\
    601--1,800              & 2                             & 2.52     & 2,857 & 14  \\
    1,801--5,400            & 3                             & 3.15     & 8,571 & 41  \\
    5,401--16,200           & 4                             & 3.78     & 25,714 & 122 \\
    \bottomrule
  \end{tabular}
  \caption{Performance of \minnowMTB{} for different bundle sizes at
    the special configuration (see text). Throughput and machines correspond to
    fully utilized bundles (largest $\bundlesize$ in each row).}
  \label{tab:latency-throughput}
\end{table}

\subsection{Latency breakdown}
\label{sec:breakdown}

This section expands on \Cref{sec:end-to-end-evaluation} by analyzing the
individual components of the latency of \minnowMTB{}.
\Cref{fig:latency-breakdown} shows the breakdown of prover
latency for bundle sizes of $N=200$, $N=600$, and $N=1{,}800$. As predicted by the microbenchmarks in \Cref{sec:eval}, the
median latency of each prover stays at about 0.63\,s,
independent of the bundle size.

The partial witness builder (\emph{Pw builder} in the figure) is the
time the primary takes to construct the partial witness. Its latency
grows slightly with the number of workers, and thus with the bundle size,
but remains low.

The figure corroborates the scalability claims of \Cref{sec:eval}: prover
latency is invariant in the number of workers.
Consequently, any observed differences between the predicted and
measured end-to-end latencies can be attributed to additional
overheads, such as network latency, client-side queuing delays,
context switching on the primary machine, and the time required to
assemble the partial witnesses.

\end{document}